\documentclass{article}
\usepackage[utf8]{inputenc}

\usepackage{amssymb}
\usepackage{amsmath}
\usepackage{wasysym}
\usepackage{graphics}
\usepackage{setspace}
\usepackage{fancybox} 
\newtheorem{definition}{Definition}
\newtheorem{theorem}{Theorem}
\newtheorem{lemma}{Lemma}

\newtheorem{corollary}{Corollary}

\newenvironment{proof}{\noindent{\sf Proof.}}{\hfill $\boxtimes\hspace{2mm}$\linebreak}
\renewcommand{\phi}{\varphi}
\renewcommand{\epsilon}{\varepsilon}

\newcommand{\rhddash}{\rhd}

\title{Price of Influence}
\title{Reasoning about Marketing in Social Networks}
\title{Marketing in Social Networks}
\title{Marketing Impact on Diffusion in Social Networks}
\author{Pavel Naumov \and Jia Tao}
\author{Pavel Naumov$^\star$ \and  Jia Tao$^\ast$}

\begin{document}

\maketitle

{\let\thefootnote\relax\footnotetext{$^\star$ Department of Computer  Science, Illinois Wesleyan University, Bloomington, Illinois, the United States, {\sf pnaumov@iwu.edu}}

\let\thefootnote\relax\footnotetext{ $^\ast$ Department of Computer Science, The College of New Jersey, Ewing, New Jersey, the United States, {\sf taoj@tcnj.edu}}}

% TALKING POINTS
% \begin{enumerate}
% \item IJCAI(no proofs) + LOFT(one proof) + Synthese(full)?
% \item game setting: drug dealers give free samples vs anti-drug campaign.
% \end{enumerate}

\begin{abstract}
The paper proposes a way to add marketing into the standard threshold model of social networks. Within this framework, the paper studies logical properties of the influence relation between sets of agents in social networks. Two different forms of this relation are considered: one for promotional marketing and the other for preventive marketing. In each case a sound and complete logical system describing properties of the  influence relation is proposed. Both systems could be viewed as extensions of Armstrong's axioms of functional dependency from the database theory. 
\end{abstract}

\section{Introduction}

\subsection{Social Networks}

In this paper we study how diffusion in social networks could be affected by marketing. Diffusion happens when a product or a social norm is initially adopted by a small group of agents who later influence their peers to adopt the same product. The peers influence their peers, and so on. There are two most commonly used models of diffusion: the cascading model and the threshold model. In the cascading model~\cite{snk08kiies,kkt05alp} the behaviour of agents is random and the peer influence manifests itself in a change of a probability of an agent to adopt the product. In the threshold model~\cite{v96sn,macy91asr,kkt03sigkdd,am14fi}, originally introduced by Granovetter~\cite{g78ajs} and Schelling~\cite{s78}, the behavior of the agents is deterministic. 

The focus of this paper is on the threshold model of diffusion of a given product. In this model, there is a threshold value $\theta(a)$ associated with each agent $a$ and an influence value $w(a,b)$ associated with each pair of agents $a$ and $b$. Informally, the threshold value $\theta(a)$ represents the resistance of agent $a$ to adoption of the product and the influence value $w(a,b)$ represents the peer pressure that agent $a$ puts on agent $b$ upon adopting the product. If the total peer pressure from the set of agents $A$ who have already adopted the product on an agent $b$ is no less than the threshold value $\theta(b)$, i.e.,
\begin{equation}\label{intro theta}
\sum_{a\in A}w(a,b)\ge \theta(b),
\end{equation}
then agent $b$ also adopts the product.

\subsection{Influence Relation}

We say that a set of agents $A$ influences a set of agents $B$ if the social network is such that an adoption of the product by all agents in set $A$ will unavoidably lead to an adoption of the product by all agents in set $B$. Note that it is not important how original adoption of the product by agents in set $A$ happens. For example, agents in set $A$ can receive and start using free samples of the product. Also,  agents in set $A$ can influence agents in set $B$ indirectly. If agents in set $A$ put enough peer pressure on some other agents to adopt the product, who in turn put enough  peer pressure on the agents in set $B$ to adopt the product, we still say that set $A$ influences set $B$. We denote this influence relation by $A\rhd B$.

%In this paper we focus on universal principles of influence that are true for all social networks. The set of all such principles has been studied by [AN-Lighthouse], who gave a complete axiomatization of these principles by the following three axioms of influence:

In this paper we focus on universal principles of influence that are true for all social networks. The set of such principles for a fixed distribution of influence values has been studied by Azimipour and Naumov~\cite{an16arxiv}, who provided a complete axiomatization of these principles that consists of the following three axioms of influence:
\begin{enumerate}
\item Reflexivity: $A\rhd B$, where $B\subseteq A$,
\item Augmentation: $A\rhd B\to (A,C\rhd B,C)$,
\item Transitivity: $A\rhd B\to (B\rhd C \to A\rhd C)$,
\end{enumerate}
and an additional fourth axiom describing a property specific to the fixed distribution of influence values. In these axioms, $A,B$ denotes the union of sets $A$ and $B$. The three axioms above were originally proposed by Armstrong~\cite{a74} to describe functional dependence relation in database theory.  They became known in database literature as Armstrong's axioms \cite[p.~81]{guw09}. V{\"a}{\"a}n{\"a}nen proposed a first order version of these principles~\cite{v07}. Beeri, Fagin, and Howard~\cite{bfh77} suggested a variation of Armstrong's axioms that describes properties of multi-valued dependence. Naumov and Nicholls~\cite{nn14jpl} proposed another variation of these axioms that describes a rationally functional dependence.

%The main result in [AN-Lighthouse] is a complete axiomatization of all properties of influence for a given distribution of influence values. The above Armstrong's axioms is a part of this axiomatization.
%This axiomatization, in addition to Armstring's axioms above, contains a new principle that the authors call {\em Lighthouse axiom}.

There have been at least two different attempts to enrich the language of Armstring's axioms by introducing an additional parameter to the functional dependence relation. V{\"a}{\"a}n{\"a}nen~\cite{v14arxiv} studied approximate dependence relation $A\rhd_p B$, where $p$ refers to the fraction of ``exceptions" in which functional dependence does not hold. In our previous work~\cite{nt15arxiv2}, we interpreted relation $A\rhd_p B$ as ``knowing values of database attributes $A$ and having an additional budget $p$ one can reconstruct the values of attributes in set $B$". In the current paper we interpret $A\rhd_p B$ as the influence relation in social networks with parameter $p$ referring to the available marketing budget to either promote or prevent influence.  

\subsection{Marketing Impact}

We propose an extension of the threshold model that incorporates marketing. This is done by representing a marketing campaign as a non-negative spending function $s$, where $s(b)$ specifies the amount of money spent on marketing the product to agent $b$. In addition, we associate a value $\lambda(b)$ with each agent $b$, which we call the {\em propensity} of agent $b$. This value represents the resistance of agent $b$ to marketing. The higher the value of the propensity is, the more responsive the agent is to the marketing. We modify formula~(\ref{intro theta}) to say that agent $b$ adopts the product if the total sum of the marketing pressure and the peer pressure from the set of agents who have already adopted the product is no less than the threshold value:
\begin{equation}\label{intro lambda}
\lambda(b)\cdot s(b) + \sum_{a\in A}w(a,b)\ge \theta(b).
\end{equation}

In the first part of this paper we assume that the goal of marketing is to {\em promote} the adoption of the product. In the second part of the paper we investigate marketing campaigns designed to {\em prevent} adoption of the product. An example of the second type of campaign is an anti-smoking advertisement campaign. In either of these two cases, the same equation~(\ref{intro lambda}) describes the condition under which the product is adopted by agent $b$.

Note that most people would be more likely to buy a product when they are exposed to a promotional marketing campaign. That is, in case of promotional marketing, the value of the propensity is usually positive. On the other hand, people are usually less likely to buy a product or to adopt a social norm  after being exposed to preventive marketing. In other words, in the preventive marketing setting, the value of the propensity is usually negative. However, our framework is general enough to allow for the propensity value to be either positive or negative in both of these cases.

While studying the marketing that promotes adoption of the product, we interpret predicate $A\rhd_p B$ as ``there is a marketing campaign with budget no more than $p$ that guarantees that the set of agents $A$ will influence the set of agents $B$". As we show, the following three modified Armstrong's axioms give a sound and complete axiomatization of universal propositional properties of this relation:
\begin{enumerate}
\item Reflexivity: $A\rhd_p B$, where $B\subseteq A$,
\item Augmentation: $A\rhd_p B\to A,C\rhd_p B,C$,
\item Transitivity: $A\rhd_p B\to(B\rhd_q C\to A\rhd_{p+q} C)$.
\end{enumerate}
These axioms are identical to our axioms of budget-constrained functional dependence~\cite{nt15arxiv2}.

In the case of marketing that aims to {\em prevent} the influence, one would naturally be interested in considering relation  ``there is a marketing campaign with budget no more than $p$ that guarantees that the set of agents $A$ will {\em not} influence the set of agents $B$". Equivalently, one can study the properties of the negation of this relation, or, in other words, the properties of the relation ``for {\em any} preventive marketing campaign with budget no more than $p$, the set of agents $A$ is able to influence the set of agents $B$". We have chosen to study the latter relation because the axiomatic system for this relation is more elegant. In this paper we show that the following four axioms give a sound and complete axiomatization of the latter relation:

\begin{enumerate}
\item Reflexivity: $A\rhd_p B$, where $B\subseteq A$,
\item Augmentation: $A\rhd_p B\to A,C\rhd_p B,C$,
\item Transitivity: $A\rhd_p B\to(B\rhd_p C\to A\rhd_p C)$,
\item Monotonicity: $A\rhd_p B\to A\rhd_q B$, where $q\le p$.
\end{enumerate}

The difference between the axiomatic systems for promotional marketing and preventive marketing is in transitivity and monotonicity axioms. Both systems include a form of transitivity axiom, but these forms are different and not equivalent. The system for preventive marketing contains a form of monotonicity axiom. For promotional marketing, the following form of monotonicity axiom is true and provable, as is shown in Lemma~\ref{mono lemma}:
$$
A\rhd_p B\to A\rhd_q B,\mbox{ where $p\le q$}.
$$ 
Both of the above axiomatic systems differ from V{\"a}{\"a}n{\"a}nen~\cite{v14arxiv} axiomatization of approximate functional dependence. 

The paper is organized as follows. In Section~\ref{social networks section}, we give formal definitions of a social network and of a diffusion in such networks. We also prove basic properties of diffusion used later in the paper. This section of the paper is common to both promotional and preventive marketing. In Section~\ref{promotional section}, we introduce semantics of promotional marketing, give axioms of our logical system for promotional marketing and prove the soundness and  the completeness of this logical system. In Section~\ref{preventive section}, we do the same for preventive marketing. Section~\ref{conclusion section} concludes the paper.

\section{Social Networks}\label{social networks section}

As discussed in the introduction, the threshold model of a social network is specified by a non-negative influence value between any pair of agents in the network and by a threshold value for each agent. Additionally, each agent is assigned a propensity value that specifies the resistance of the agent to marketing. The value of the propensity could be positive, zero, or negative. We assume that the set of agents is finite.  

\begin{definition}\label{social network definition}
A social network is a tuple $(\mathcal{A}, w, \lambda, \theta)$, where
\begin{enumerate}
\item Set $\mathcal{A}$ is a finite set of agents.
\item Function $w$ maps $\mathcal{A}\times \mathcal{A}$ into the set of non-negative real numbers. The value $w(a,b)$ represents the ``influence" of agent $a$ on agent $b$.
\item Function $\lambda$ maps $\mathcal{A}$ into real numbers. The value of $\lambda(a)$ represents the ``propensity" of an agent $a$ to marketing.
\item ``Threshold" function $\theta$ maps $\mathcal{A}$ into the set of real numbers. 
\end{enumerate}
\end{definition}

Figure~\ref{social network figure} illustrates Definition~\ref{social network definition}. In this figure, the set of agents is the set $\{u,v,w,t,x,y,z\}$. The influence value $w(a,b)$ is specified by the label on the directed edge from $a$ to $b$. The edges for which the influence value is zero are omitted. Threshold and propensity values are shown next to each agent.

\begin{figure}[ht]
\begin{center}
\vspace{0mm}
\scalebox{.6}{\includegraphics{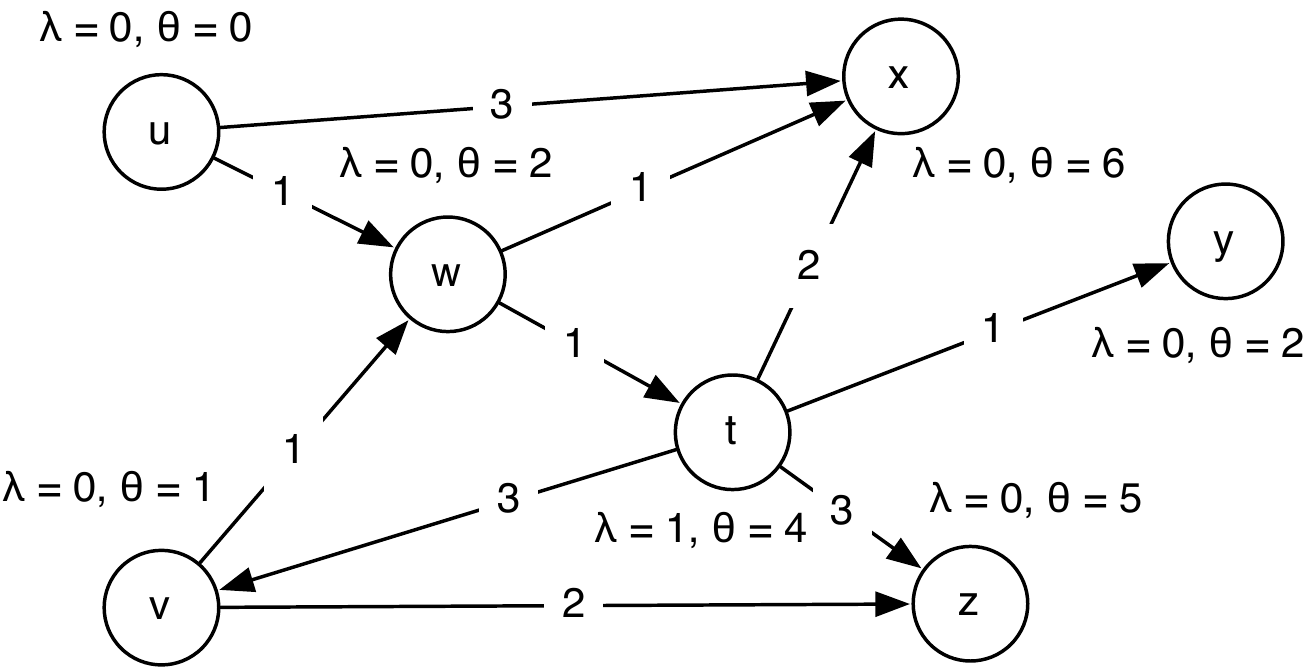}}
\vspace{0mm}
\footnotesize
\caption{A Social Network.}\label{social network figure}
\vspace{0cm}
\end{center}
\vspace{0mm}
\end{figure}

We describe a marketing campaign by specifying ``spending" on advertisement to each agent in the social network. 
\begin{definition}\label{spending function}
For any social network $(\mathcal{A}, w, \lambda, \theta)$, a spending function is an arbitrary function from set $\mathcal{A}$ into non-negative real numbers.
\end{definition}
The following is an example of a spending function for the social network depicted in Figure~\ref{social network figure}. This function specifies a marketing campaign targeting exclusively agent $t$.
\begin{equation}\label{sample s}
s(a)=
\begin{cases}
3, & \mbox{if $a=t$,}\\
0, & \mbox{otherwise}.
\end{cases}
\end{equation}
\begin{definition}\label{||set}
For any social network $(\mathcal{A}, w, \lambda, \theta)$ and any spending function $s$, let
$\|s\| = \sum_{a\in\mathcal{A}}s(a)$.
\end{definition}
For the spending function defined by equation~(\ref{sample s}), we have $\|s\|=3$. 

Next we formally define the diffusion in social network under marketing campaign specified by a spending function $s$. Suppose that initially the product is adopted by a set of agent $A$. We recursively define the {\em diffusion chain}  of sets of agents
$$
A=A^0_s\subseteq A^1_s\subseteq A^2_s\subseteq A^3_s\subseteq \dots,
$$
where $A^k_s$ is the set of agents who have adopted the product on or before the $k$-th step of the diffusion.

\begin{definition}\label{Ak}
For any given social network $(\mathcal{A}, w, \lambda, \theta)$, any spending function $s$, and any subset $A\subseteq \mathcal{A}$, let
set $A^n_s$ be recursively defined as follows:
\begin{enumerate}
\item $A^0_s=A$,
\item $A^{n+1}_s=A^n_s\cup\left\{b\in \mathcal{A}\;|\;  \lambda(b)\cdot s(b) + \sum_{a\in A^n_s} w(a,b) \ge \theta(b) \right\}$. 
\end{enumerate}
\end{definition}

For example, consider again the social network depicted in Figure~\ref{social network figure}. Let $A$ be the set $\{v\}$ and $s$ be the spending function defined by equation~(\ref{sample s}). Note that the threshold value of agent $u$ in this network is zero and, thus, it will adopt the product without any peer or marketing pressure. For the other agents in this network, the combination of the marketing pressure specified by the marketing function $s$ and the peer pressure from agent $v$ is not enough to adopt the product. Thus, $A^1_s=\{v,u\}$. Once agent $v$ and agent $u$ both adopt the product, their combined peer pressure on agent $w$ reaches the threshold value of $w$ and agent $w$ also adopts the product. No other agent is experiencing enough pressure to adopt the product at this point. Hence, $A^2_s=\{v,u,w\}$. Next, agent $t$ will adopt the product due to the combination of the peer pressure from agent $w$ and the marketing pressure specified by the spending function $s$, and so on. This diffusion process is illustrated in Figure~\ref{social network Ak figure}.

\begin{figure}[ht]
\begin{center}
\vspace{0mm}
\scalebox{.6}{\includegraphics{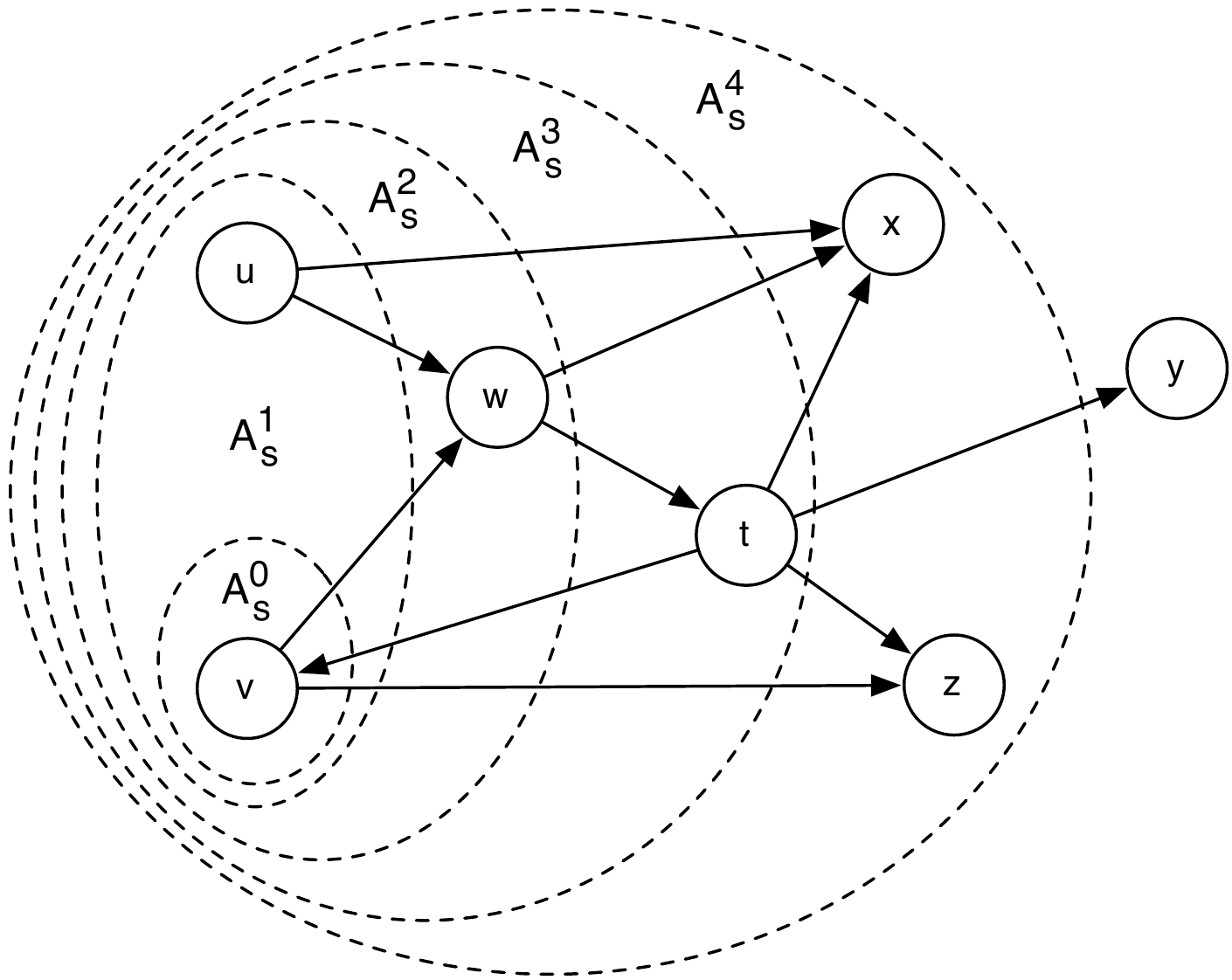}}
\vspace{0mm}
\footnotesize
\caption{Diffusion Chain $A^1_s\subseteq A^2_s\subseteq A^2_s\subseteq A^4_s$.}\label{social network Ak figure}
\vspace{0cm}
\end{center}
\vspace{0mm}
\end{figure}

\begin{corollary}\label{An+k}
$(A^n_s)^k_s=A^{n+k}_s$ for each social network $(\mathcal{A}, w, \lambda, \theta)$, each $n,k\ge 0$, each set $A\subseteq\mathcal{A}$, and each spending function $s$.
\end{corollary}

\begin{definition}\label{A*}
$A^*_s=\bigcup_{n\ge 0}A^n_s$.
\end{definition}

\begin{corollary}\label{A subseteq A*s}
$A\subseteq A^*_s$ for each social network $(\mathcal{A}, w, \lambda, \theta)$, each spending function $s$, and each subset $A\subseteq \mathcal{A}$.
\end{corollary}

In the rest of this section we establish technical properties of the chain $\{A^n_s\}_{n\ge 0}$ and the set $A^*_s$ that are used later.
The first of these properties is a corollary that follows from the assumption of the finiteness of set $\mathcal{A}$ in Definition~\ref{social network definition}.
\begin{corollary}\label{A*=Ak}
For any social network $(\mathcal{A}, w, \lambda, \theta)$, any subset $A$ of $\mathcal{A}$ and any
spending function $s$, there is $n\ge 0$ such that $A^*_s=A^n_s$.
\end{corollary}

%\section{Properties of $A^n_s$ and $A^*_s$}

Next, we prove that $A^*_s$ is an idempotent operator. 
\begin{lemma}\label{idempotent}
$(A^*_s)^*_s\subseteq A^*_s$ for each social network $(\mathcal{A}, w, \lambda, \theta)$, each spending function $s$, and each subsets $A$ of $\mathcal{A}$.
\end{lemma}
\begin{proof}
By Corollary~\ref{A*=Ak}, there is $n\ge 0$ such that $A^*_s=A^n_s$. By the same corollary, there also is $k\ge 0$ such that $(A^n_s)^*_s= (A^n_s)^k_s$. Thus, by Corollary~\ref{An+k},
$$
(A^*_s)^*_s=(A^n_s)^*_s=(A^n_s)^k_s=A^{n+k}_s.
$$
Therefore, $(A^*_s)^*_s\subseteq A^*_s$ by Definition~\ref{A*}.
\end{proof}

We now show that any set of agents influences at least as many agents as any of its subsets, given the same fixed spending function. This claim is formally stated as Corollary~\ref{star mono} that follows from the next lemma:

\begin{lemma}\label{subset Ak}
If $A\subseteq B$, then $A^k_s\subseteq B^k_s$, for each social network $(\mathcal{A}, w, \lambda, \theta)$, each spending function $s$, each $k\ge 0$, and all subsets $A$ and $B$ of $\mathcal{A}$.
\end{lemma}
\begin{proof}
We prove the statement of the lemma by induction on $k$. If $k=0$, then $A^0_s=A\subseteq B=B^0_s$ by Definition~\ref{Ak}.

Suppose that $A^k_s\subseteq B^k_s$. Let $x\in A^{k+1}_s$. It suffices to show that $x\in B^{k+1}_s$. Indeed, by Definition~\ref{Ak}, assumption $x\in A^{k+1}_s$ implies that either $x\in A^{k}_s$ or $\lambda(x)\cdot s(x) + \sum_{a\in A^{k}_s}w(a,x)  \ge \theta(x)$. When $x\in A^{k}_s$, by the induction hypothesis, $x\in A^{k}_s\subseteq B^{k}_s$. Thus, $x\in B^{k}_s$. Therefore, $x\in B^{k+1}_s$ by Definition~\ref{Ak}.

When $\lambda(x)\cdot s(x) + \sum_{a\in A^{k}_s}w(a,x)  \ge \theta(x)$, due to the assumption $A^k\subseteq B^k$,
$$
\lambda(x)\cdot s(x) + \sum_{b\in B^k_s} w(b,x)  \ge \lambda(x)\cdot s(x) + \sum_{a\in A^k_s} w(a,x)  \ge \theta(x).
$$
Therefore, $x\in B^{k+1}_s$ by Definition~\ref{Ak}.
\end{proof}

\begin{corollary}\label{star mono}
If $A\subseteq B$, then $A^*_s\subseteq B^*_s$, for each social network $(\mathcal{A}, w, \lambda, \theta)$, each spending function $s$, and all subsets $A$ and $B$ of $\mathcal{A}$.
\end{corollary}
Next, we establish that the influence of the union of two sets of agents is at least as strong as the combination of the influence of these two sets.

\begin{lemma}\label{star union}
$A^*_s\cup B^*_s\subseteq (A\cup B)^*_s$, for each social network $(\mathcal{A}, w, \lambda, \theta)$, each spending function $s$, and all subsets $A$ and $B$ of $\mathcal{A}$.
\end{lemma}
\begin{proof}
Note that $A\subseteq A\cup B$ and $B\subseteq A\cup B$. Thus, $A^*_s\subseteq (A\cup B)^*_s$ and $B^*_s\subseteq (A\cup B)^*_s$ by Corollary~\ref{star mono}. Therefore, $A^*_s\cup B^*_s\subseteq (A\cup B)^*_s$.
\end{proof}

One might intuitively think that the result of two consecutive marketing campaigns can not be more effective than the combined campaign, or, in other terms, that $(A^*_{s_1})^*_{s_2}\subseteq A^*_{s_1 + s_2}$. More careful analysis shows that this claim is true only if all agents have non-negative propensity. However, this property can be restated in the form which is true for negative propensity as well. To do this, we introduce a binary operation $\oplus_\lambda$ on spending functions. 

\begin{definition}\label{oplus}
For any two spending functions $s_1$ and $s_2$ and any propensity function $\lambda$, let $s_1\oplus_\lambda s_2$ be spending function such that for each agent $a$,
$$
(s_1\oplus_\lambda s_2)(a)=
\begin{cases}
s_1(a)+s_2(a), & \mbox{ if $\lambda(a)\ge 0$},\\
0, & \mbox{ otherwise.}
\end{cases}
$$
\end{definition}

The desired property, expressed in terms of operation $\oplus_\lambda$, is stated later as Lemma~\ref{A**}. We start with two auxiliary observations.

\begin{lemma}\label{lambda oplus}
$
\lambda(b)\cdot s_1(b)\le 
\lambda(b)\cdot(s_1\oplus_\lambda s_2)(b)
$
for any social network $(\mathcal{A},w,\lambda,\theta)$, any
agent $b\in\mathcal{A}$, and any two spending functions $s_1$ and $s_2$.
\end{lemma}
\begin{proof} We consider the following two cases separately:

\noindent{\em Case I:} $\lambda(b)\ge 0$. In this case by Definition~\ref{oplus} and because $s_2(b)\ge 0$ due to Definition~\ref{spending function}, we have
$s_1(b)\le s_1(b) + s_2(b)=  (s_1\oplus_\lambda s_2)(b)$. Therefore, 
$
\lambda(b)\cdot s_1(b)\le 
\lambda(b)\cdot(s_1\oplus_\lambda s_2)(b)
$
 by the assumption $\lambda(b)\ge 0$.

\noindent{\em Case II:} $\lambda(b)< 0$. In this case by Definition~\ref{oplus} and because $s_1(b)\ge 0$ due to Definition~\ref{spending function}, we have
$s_1(b) \ge 0 =  (s_1\oplus_\lambda s_2)(b)$. Therefore, 
$
\lambda(b)\cdot s_1(b)\le 
\lambda(b)\cdot(s_1\oplus_\lambda s_2)(b)
$ by the assumption $\lambda(b)< 0$.
\end{proof}

Now we prove that the spending function $s_1\oplus_\lambda s_2$ is at least as effective as $s_1$.
\begin{lemma}\label{s1s2}
$A^n_{s_1} \subseteq  A^n_{s_1\oplus_\lambda s_2}$, for any social network $(\mathcal{A},w,\lambda,\theta)$, any set $A\subseteq \mathcal{A}$, any $n\ge 0$, any propensity function $\lambda$, and any two spending function $s_1$ and $s_2$.
\end{lemma}
\begin{proof}
We show the lemma by induction on $n$. If $n=0$, then, by Definition~\ref{Ak},
$A^0_{s_1}=A=A^0_{s_1\oplus_\lambda s_2}$. Suppose that $A^n_{s_1}\subseteq A^n_{s_1\oplus_\lambda s_2}$. We need to show that $A^{n+1}_{s_1}\subseteq A^{n+1}_{s_1\oplus_\lambda s_2}$. Indeed, by Definition~\ref{Ak}, Lemma~\ref{lambda oplus}, and the induction hypothesis,
\begin{eqnarray*}
A^{n+1}_{s_1} &=& A^n_{s_1}\cup\left\{b\in \mathcal{A}\;|\;  \lambda(b)\cdot s_1(b) + \sum_{a\in A^n_{s_1}} w(a,b) \ge \theta(b) \right\}\\
&\subseteq& A^n_{s_1\oplus_\lambda s_2}\cup\left\{b\in \mathcal{A}\;|\;  \lambda(b)\cdot (s_1\oplus_\lambda s_2)(b) + \sum_{a\in A^n_{s_1\oplus_\lambda s_2}} w(a,b) \ge \theta(b) \right\} \\
&=& A^{n+1}_{s_1\oplus_\lambda s_2}.
\end{eqnarray*}
\end{proof}

Finally, we are ready to state and prove that a marketing campaign with spending function $s_1\oplus_\lambda s_2$ is at least as effective as a sequential combination of two marketing campaigns with spending functions $s_1$ and $s_2$. This property is used in Lemma~\ref{promotional transitivity sound} to prove the soundness of Transitivity axiom for promotional marketing.
\begin{lemma}\label{A**}
$(A^*_{s_1})^*_{s_2}\subseteq A^*_{s_1\oplus_\lambda s_2}$, for any social network $(\mathcal{A},w,\lambda,\theta)$, any set $A\subseteq \mathcal{A}$, any propensity function $\lambda$, and any two spending function $s_1$ and $s_2$.
\end{lemma}
\begin{proof}
By Corollary~\ref{A*=Ak}, there are $n_1,n_2\ge 0$ such that $A^*_{s_1}=A^{n_1}_{s_1}$ and $(A^{n_1}_{s_1})^*_{s_2}=(A^{n_1}_{s_1})^{n_2}_{s_2}$. Thus,

\[\arraycolsep=2pt\def\arraystretch{1.6}
\begin{array}{rclcr}
(A^*_{s_1})^*_{s_2} & = & (A^{n_1}_{s_1})^{n_2}_{s_2}\\
& \subseteq  & (A^{n_1}_{s_1\oplus_\lambda s_2})^{n_2}_{s_2} &\phantom{AAA} & \mbox{by Lemma~\ref{s1s2} and Lemma~\ref{subset Ak}}\\
& \subseteq & (A^{n_1}_{s_1\oplus_\lambda s_2})^{n_2}_{s_2\oplus_\lambda s_1} & & \mbox{by Lemma~\ref{s1s2}}\\
& \subseteq & (A^{n_1}_{s_1\oplus_\lambda s_2})^{n_2}_{s_1\oplus_\lambda s_2} & & \mbox{by Definition~\ref{oplus}}\\
& \subseteq & A^{n_1+n_2}_{s_1\oplus_\lambda s_2} & & \mbox{by Corollary~\ref{An+k}}\\
& \subseteq & A^{*}_{s_1\oplus_\lambda s_2} & & \mbox{by Definition~\ref{A*}.}
\end{array}
\]
\end{proof}

\section{Logic of Promotional Marketing}\label{promotional section}

There are two logical systems that we study in this paper. In this section we introduce a logical system for the marketing aiming to promote influence and prove its soundness and completeness. In the next section we do the same for the marketing aiming to prevent influence. 

\subsection{Syntax and Semantics}

We start by defining the syntax of our logical systems. The logic of promotional marketing and the logic of preventive marketing use the same language $\Phi(\mathcal{A})$, but different semantics.

\begin{definition}
For any finite set $\mathcal{A}$, let $\Phi(\mathcal{A})$ be the minimum set of formulas such that
\begin{enumerate}
\item $A\rhd_p B\in \Phi(\mathcal{A})$ for all subsets $A$ and $B$ of set $\mathcal{A}$ and all non-negative real numbers $p$,
\item $\neg\phi\in\Phi(\mathcal{A})$ for all $\phi\in \Phi(\mathcal{A})$,
\item $\phi\to\psi\in\Phi(\mathcal{A})$ for all $\phi,\psi\in\Phi(\mathcal{A})$.
\end{enumerate}
\end{definition}

The next definition is the key definition of this section. Its item 1 specifies the influence relation in a social network with a  fixed marketing budget.
\begin{definition}\label{sat}
For any social network $N$ with the set of agents $\mathcal{A}$ and any formula $\phi\in\Phi(\mathcal{A})$, we define satisfiability relation $N\vDash\phi$ as follows:
\begin{enumerate}
\item $N\vDash A\rhd_p B$ if $B\subseteq A^*_s$ for some spending function $s$ such that $\|s\|\le p$,
\item $N\vDash\neg\psi$ if $N\nvDash\psi$,
\item $N\vDash\psi\to\chi$ if $N\nvDash\psi$ or $N\vDash\chi$.
\end{enumerate}
\end{definition}
For example, as we have seen in the introduction, for social network $N$ depicted in Figure~\ref{social network figure}, we have $\{x,z\}\subseteq\{v\}^*_s$, where spending function $s$ is defined by equation~(\ref{sample s}). Thus, $N\vDash \{v\}\rhd_3 \{x,z\}$. Through the rest of the paper we omit curly braces from the formulas like this and write them simply as $N\vDash v\rhd_3 x,z$.

\subsection{Axioms}
Let $\mathcal{A}$ be any fixed finite set of agents. Our logical system for promotional influence, in addition to propositional tautologies in language $\Phi(\mathcal{A})$, contains the following axioms:

\begin{enumerate}
\item Reflexivity: $A\rhd_p B$, where $B\subseteq A$,
\item Augmentation: $A\rhd_p B\to A,C\rhd_p B,C$,
\item Transitivity: $A\rhd_p B\to(B\rhd_q C\to A\rhd_{p+q} C)$.
\end{enumerate}

We write $\vdash \phi$ if formula $\phi\in\Phi(\mathcal{A})$ is derivable in this logical system using Modus Ponens inference rule. We write $X\vdash \phi$ if formula $\phi$ is derivable using an additional set of axioms $X\subseteq \Phi(\mathcal{A})$.

\subsection{Examples}

The soundness and the completeness of our logical system will be shown later. In this section we give several examples of formal proofs in our system. We start with a form of the monotonicity statement from the introduction. As the next lemma shows, this statement is provable in our logic of promotional marketing when $p\le q$:

\begin{lemma}\label{mono lemma}
$\vdash A\rhd_p B\to A\rhd_q B$, where $p\le q$.
\end{lemma}
\begin{proof}
By Transitivity axiom,
$
\vdash A\rhd_{q-p} A \to(A\rhd_p B\to A\rhd_q B).
$
At the same time, $\vdash A\rhd_{q-p} A$ by Reflexivity axiom. Thus, $\vdash A\rhd_p B\to A\rhd_q B$ by Modus Ponens inference rule.
\end{proof}

\begin{lemma}\label{union lemma}
$\vdash A\rhd_p B \to (A\rhd_q C \to A\rhd_{p+q} B,C)$. 
\end{lemma}
\begin{proof}
By Augmentation axiom,
\begin{equation}\label{eq 1}
\vdash A\rhd_p B \to A\rhd_p A,B
\end{equation}
and 
\begin{equation}\label{eq 2}
\vdash A\rhd_q C \to A,B\rhd_q B,C.
\end{equation}
By Transitivity axiom,
\begin{equation}\label{eq 3}
\vdash A\rhd_p A,B \to (A,B\rhd_q B,C \to A\rhd_{p+q} B,C).
\end{equation}
The statement of the lemma follows from statements (\ref{eq 1}), (\ref{eq 2}), and (\ref{eq 3}) by the laws of propositional logic.
\end{proof}

The next lemma will be used later in the proof of the completeness.

\begin{lemma}\label{big union lemma}
Let $X$ be a subset of $\Phi(\mathcal{A})$, $m$ be a non-negative integer number, sets $A,B_1,\dots,B_m$ be subsets of $\mathcal{A}$, and $p_1,\dots,p_m$ be non-negative real numbers.
If $X\vdash A\rhd_{p_i} B_i$ for all $1\le i\le m$, then $X\vdash A\rhd_{q} \bigcup_{i=1}^m B_i$, where $q=\sum_{i=1}^m p_i$. 
\end{lemma}
\begin{proof}
We prove the lemma by induction on $m$. If $m=0$, then we need to show that $X\vdash A\rhd_0\varnothing$, which is an instance of Reflexivity axiom.

Suppose that $X\vdash A\rhd_{q'} \bigcup_{i=1}^{m-1}B_i$, where $q'=\sum_{i=1}^{m-1} p_i$. Since $X\vdash A\rhd_{p_m} B_m$ due to the assumption of the lemma, by Lemma~\ref{union lemma}, $X\vdash A\rhd_{q} \bigcup_{i=1}^m B_i$, where $q=\sum_{i=1}^m p_i$.
\end{proof}

\subsection{Soundness}

In this section we prove the soundness of the logic for promotional marketing. 

\begin{theorem}\label{promotional soundness}
For any finite set $\mathcal{A}$ and any $\phi\in\Phi(\mathcal{A})$,
if $\vdash \phi$, then $N\vDash\phi$ for each social network $N=(\mathcal{A},w,\lambda,\theta)$.
\end{theorem}

The soundness of propositional tautologies and of Modus Ponens inference rule is straightforward. Below we show the soundness of each of the remaining axioms as a separate lemma.

\begin{lemma}
$N\vDash A\rhd_p B$, for any social network $N=(\mathcal{A}, w, \lambda, \theta)$ and any subsets $A$ and $B$ of $\mathcal{A}$ such that $B\subseteq A$.
\end{lemma}
\begin{proof} Let $s$ be the spending function equal to 0 on each $a\in \mathcal{A}$. Thus, $\|s\|=0\le p$ by Definition~\ref{||set}. At the same time,
$B\subseteq A  \subseteq  A^*_s$ by Corollary~\ref{A subseteq A*s}. Therefore, $N\vDash A\rhd_p B$ by Definition~\ref{sat}.
\end{proof}

\begin{lemma}
If $N\vDash A\rhd_p B$, then $N\vDash A,C\rhd_p B,C$, for each social network $N=(\mathcal{A}, w, \lambda, \theta)$ and all subsets $A$, $B$, and $C$ of $\mathcal{A}$.
\end{lemma}
\begin{proof}
Suppose that $N \vDash A \rhd_p B $. Thus, by Definition~\ref{sat}, there is a spending function $s$ such that $\|s\|\le p$ and $B\subseteq A^*_s$ . Note that $C\subseteq C^*_s$ by Corollary~\ref{A subseteq A*s}. Thus, $B\cup C \subseteq A^*_s \cup C^*_s \subseteq (A\cup C)^*_s$ by Lemma~\ref{star union}.
Therefore, $N \vDash A,C \rhd_p B,C $, by Definition~\ref{sat}.
\end{proof}

\begin{lemma}\label{promotional transitivity sound} 
For any social network $N=(\mathcal{A},w,\lambda,\theta)$,
if $N \vDash A \rhd_p B$ and $N\vDash B \rhd_q C$, then $N \vDash A \rhd_{p+q} C$. 
\end{lemma}
\begin{proof}
By Definition~\ref{sat}, assumption $N\vDash B \rhd_q C$ implies that there is a spending function $s_1$ such that $\|s_1\|\le q$ and $C\subseteq B^*_{s_1}$. 

Similarly, assumption $N\vDash A \rhd_p B$  implies that there is a spending function $s_2$ such that $\|s_2\|\le p$ and $B\subseteq A^*_{s_2}$. Hence, $B^*_{s_1}\subseteq (A^*_{s_2})^*_{s_1}$ by Corollary~\ref{star mono}. Thus, $B^*_{s_1}\subseteq A^*_{s_1\oplus_\lambda s_2}$ by Lemma~\ref{A**}. 

It follows that $C\subseteq B^*_{s_1}\subseteq A^*_{s_1\oplus_\lambda s_2}$. At the same time, $\|s_1\oplus_\lambda s_2\|\le \|s_1\|+\|s_2\|\le p+q$, by Definition~\ref{oplus}. Therefore, $N \vDash A \rhd_{p+q} C$ by Definition~\ref{sat}.
\end{proof}

This concludes the proof of the soundness of our logical system for promotional marketing.

\subsection{Completeness}

We now show the completeness of our logical system for promotional marketing. This result is formally stated as Theorem~\ref{promotional completeness} in the end of this section. As usual, at the core of the proof of the completeness is a construction of a canonical model. In our case, the role of a canonical model is played by the canonical social network.

Let $\mathcal{A}_0$ be any finite set and $X=\{A_i\rhd_{p_i} B_i\}_{i\le m}$ be any finite set of atomic formulas in language $\Phi(\mathcal{A}_0)$. We now proceed to define the canonical social network $N_X=(\mathcal{A}, w,\lambda, \theta)$. An example of the canonical network for set $X$ consisting of formula $a,c\rhd_1 d$, formula $b,c\rhd_2 a$, and formula $a,b\rhd_3 c$ is depicted in Figure~\ref{canonical network figure}.
\begin{figure}[ht]
\begin{center}
\vspace{0mm}
\scalebox{.5}{\includegraphics{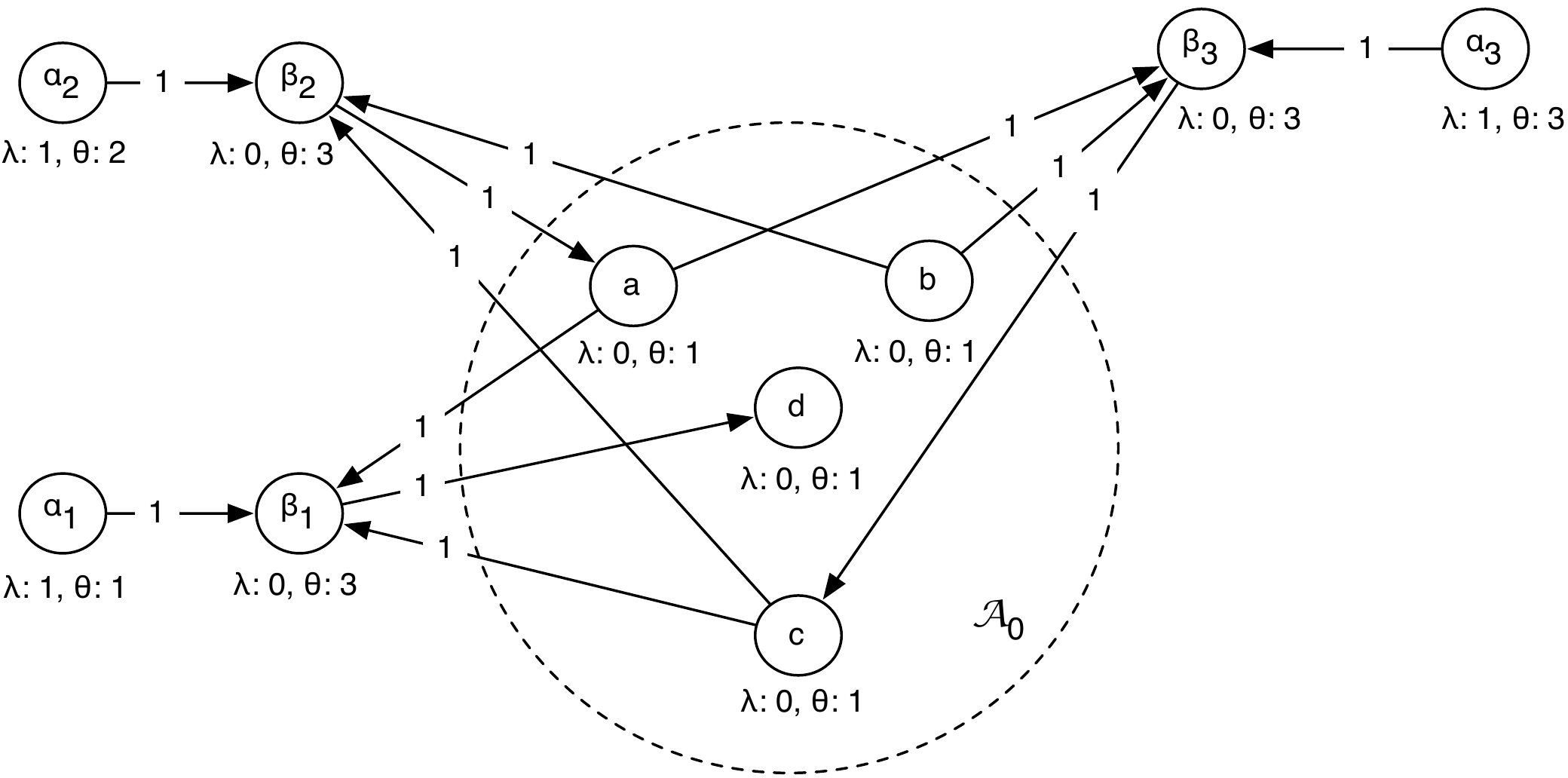}}
\vspace{0mm}
\footnotesize
\caption{The canonical social network $N_X$ for set $X$ consisting of formula $a,c\rhd_1 d$, formula $b,c\rhd_2 a$, and formula $a,b\rhd_3 c$.}\label{canonical network figure}
\vspace{0cm}
\end{center}
\vspace{0mm}
\end{figure}
We associate two new agents $\alpha_i$ and $\beta_i$ with each formula $A_i\rhd_{p_i} B_i\in X$. We assume that agents $\alpha_1,\dots,\alpha_m,\beta_1,\dots,\beta_m$ are distinct and that they do not belong to set $\mathcal{A}_0$.

\begin{definition}
$\mathcal{A}=\mathcal{A}_0 \cup \{\alpha_i\}_{i\le m} \cup \{\beta_i\}_{i\le m}$.
\end{definition}

In social network $N_X$ only agents $\{\alpha_i\}_{i\le m}$ are responsive to promotional marketing. We formally capture this through the following definition of  function $\lambda$:

\begin{definition}\label{canonical lambda} For any $a\in\mathcal{A}$,
$$
\lambda(a)=
\begin{cases}
1, & \mbox{if $a=\alpha_i$ for some $i\le m$},\\
0, & \mbox{otherwise.}
\end{cases}
$$
\end{definition}

We assume that for each $i\le m$, all agents in set $A_i$ as well as agent $\alpha_i$ put peer pressure on agent $\beta_i$ once they adopt the product. In addition, upon adopting the product, agent $\beta_i$ puts peer pressure on each agent in set $B_i$. Besides that, no agent can put peer pressure on any other agent in this network. We formally capture this in Definition~\ref{canonical w}.

\begin{definition}\label{canonical w}
$$
w(a,b)=
\begin{cases}
1, & \mbox{if $a\in A_i\cup\{\alpha_i\}$ and $b=\beta_i$ for some $i\le m$},\\
1, & \mbox{if $a=\beta_i$ and $b\in B_i$ for some $i\le m$},\\
0, & \mbox{otherwise.}
\end{cases}
$$
\end{definition}

Before continuing with the definition of the social network $N_X$, we state and prove a property of this network that follows from Definition~\ref{canonical w}. We show that in order to put peer pressure of at least $|A_i|+1$ on agent $\beta_i$, one needs to influence agent $\alpha_i$ and all of the agents in set $A_i$.

\begin{lemma}\label{canonical w lemma}
If $\sum_{a\in A^{n}_s}w(a,\beta_i)\ge |A_i|+1$, then $\alpha_i\in A^n_s$ and $A_i\subseteq A^n_s$.
\end{lemma}
\begin{proof}
By Definition~\ref{canonical w},  $w(a,\beta_i)=1$  if $a\in A_i\cup\{\alpha_i\}$ and $w(a,\beta_i)=0$ for all other $a\in\mathcal{A}$. Thus, inequality $\sum_{a\in A^{n}_s}w(a,\beta_i)\ge |A_i|+1$ implies that $A_i\cup\{\alpha_i\}\subseteq A^{n}_s$. Therefore, $\alpha_i\in A^n_s$ and $A_i\subseteq A^n_s$.
\end{proof}

We are now ready to define the threshold value function $\theta$ for the social network $N_X$. Recall that according to Definition~\ref{canonical w} and Definition~\ref{canonical lambda}, no agent can put peer pressure on agent $\alpha_i$, but agent $\alpha_i$ is responsive to promotional marketing. We set the threshold value $\theta(\alpha_i)$ to $p_i$ so that this agent can only be influenced by a marketing campaign with budget at least $p_i$. We set value $\theta(\beta_i)$ high enough to guarantee (see Lemma~\ref{canonical w lemma}) that agent $\alpha_i$ and each agent in set $A_i$ adopt the product before agent $\beta_i$ is influenced. Threshold values of all agents in set $\mathcal{A}_0$ are set to 1.

\begin{definition}\label{canonical theta}
$$
\theta(a)=
\begin{cases}
p_i, & \mbox{if $a=\alpha_i$ for some $i\le m$},\\
|A_i|+1, & \mbox{if $a=\beta_i$ for some $i\le m$},\\
1, & \mbox{otherwise.}
\end{cases}
$$
\end{definition}
This concludes the definition of the canonical social network $N_X=(\mathcal{A}, w,\lambda, \theta)$.

Recall that Figure~\ref{canonical network figure} depicts the canonical social network $N_X$ for set $X$ consisting of formula $a,c\rhd_1 d$, formula $b,c\rhd_2 a$, and formula $a,b\rhd_3 c$.
Note that formula $a,c\rhd_1 d$, formula $b,c\rhd_2 a$, and formula $a,b\rhd_3 c$ are all satisfied in the canonical network depicted in Figure~\ref{canonical network figure}. For example, for the formula $a,c\rhd_1 d$, let spending function $s$ be such that it spends 1 on agent $\alpha_i$ and nothing on all other agents. Thus, $\{a,c\}^1_s=\{a,c,\alpha_1\}$. Once $\alpha_1$ adopts the product, the total peer pressure on agent $\beta_1$ becomes 3 and it too adopts the product: 
$\{a,c\}^2_s=\{a,c,\alpha_1,\beta_1\}$. Finally, upon adopting of the product, agent $\beta_1$ alone puts enough pressure on agent $d$ to also adopt the product: $\{a,c\}^3_s=\{a,c,\alpha_1,\beta_1,d\}$. Thus, formula $a,c\rhd_1 d$ is satisfied in this network.

The next lemma generalizes the observation made in the previous paragraph to a claim that all formulas from set $X=\{A_i\rhd_{p_i}B_i\}_{i\le m}$ are satisfied in the canonical network $N_X$.

\begin{lemma}\label{first base case}
$N_X\vDash A_i\rhd_{p_i} B_i$ for each $i\le m$.
\end{lemma}
\begin{proof}
Consider any $i\le m$. Let $s$ be a spending function such that
\begin{equation}\label{canonical s}
s(a)=
\begin{cases}
p_i, & \mbox{if $a=\alpha_i$},\\
0, & \mbox{otherwise.}
\end{cases}
\end{equation}
Then, by Definition~\ref{canonical lambda}, Definition~\ref{canonical w}, and Definition~\ref{canonical theta},
$$
\lambda(\alpha_i)\cdot s(\alpha_i) + \sum_{a\in A_i} w(a,\alpha_i) 
= 1 \cdot p_i + \sum_{a\in A_i} 0= p_i = \theta(\alpha_i).
$$
Thus, $\alpha_i\in (A_i)^1_s$ by Definition~\ref{Ak}. Hence, by Definition~\ref{canonical lambda}, Definition~\ref{canonical w}, and Definition~\ref{canonical theta},
\begin{eqnarray*}
\lambda(\beta_i)\cdot s(\beta_i) + \sum_{a\in (A_i)^1_s} w(a,\beta_i) & \ge & 
\lambda(\beta_i)\cdot s(\beta_i) + w(\alpha_i,\beta_i) + \sum_{a\in A_i} w(a,\beta_i)\\
& \ge & 0 \cdot 0 + 1 + |A_i| = 1 + |A_i| = \theta(\beta_i).
\end{eqnarray*}
Thus, $\beta_i\in (A_i)^2_s$ by Definition~\ref{Ak}. Finally, for each $b\in B_i$, by Definition~\ref{canonical lambda}, Definition~\ref{canonical w}, and Definition~\ref{canonical theta},
\begin{eqnarray*}
\lambda(b)\cdot s(b) + \sum_{a\in (A_i)^2_s} w(a,b) & \ge & 
0\cdot 0 + w(\beta_i,b) = w(\beta_i,b) = 1= \theta(b).
\end{eqnarray*}
Hence, $b\in (A_i)^3_s$  by Definition~\ref{Ak}. Thus, $b\in (A_i)^*_s$  by Definition~\ref{A*} for each $b\in B_i$. Then, $B_i\subseteq A^*_s$. Note that $\|s\|=p_i$ due to definition (\ref{canonical s}). Therefore, $N_X\vDash A_i\rhd_{p_i}B_i$ by Definition~\ref{sat}.
\end{proof}

Our next important result is the converse of Lemma~\ref{first base case} stated later as Lemma~\ref{second base case}. In preparation for its, we make several technical observations about the social network $N_X$. First, we prove that, for each $i\le m$, agent $\beta_i$ can not be influenced without agent $\alpha_i$ being influenced as well.

\begin{lemma}\label{if beta then alpha}
If $\beta_i\in A^n_s$, then $\alpha_i\in A^n_s$, for each $A\subseteq \mathcal{A}_0$, each $i\le m$, and each $n\ge 0$.
\end{lemma}
\begin{proof}
Suppose $\beta_i\in A^n_s$. Let $k$ be the smallest integer  such that $0\le k\le n$ and  $\beta_i\in A^k_s$. 

If $k=0$, then $\beta_i\in A^0_s=A$ by Definition~\ref{Ak}. Thus, $\beta_i\in \mathcal{A}_0$ due to the assumption $A\subseteq \mathcal{A}_0$, which contradicts the choice of $\beta_1,\dots,\beta_m$. Therefore, the lemma is vacuously true.

Suppose that $k>0$. Since $k>0$ is the smallest integer such that $\beta_i\in A^k_s$, it must be that $\beta_i\in A^k_s\setminus A^{k-1}_s$. Thus, by Definition~\ref{Ak},
$$
\lambda(\beta_i)\cdot s(\beta_i) + \sum_{a\in A^{k-1}_s}w(a,\beta_i) \ge \theta(\beta_i).
$$
By Definition~\ref{canonical lambda}, $\lambda(\beta_i)=0$. By Definition~\ref{canonical theta}, $\theta(\beta_i)=|A_i|+1$. Thus,
$$
\sum_{a\in A^{k-1}_s}w(a,\beta_i) \ge |A_i|+1.
$$
Thus, $\alpha_i\in A^{k-1}_s$ by Lemma~\ref{canonical w lemma}. Hence, $\alpha_i\in A^{k-1}_s$. Therefore, $\alpha_i\in A^{n}_s$ by Definition~\ref{Ak} and since $k-1<k\le n$. 
\end{proof}

The next lemma shows that the only way to influence agent $\alpha_i$ is to spend at least $p_i$ on promotional marketing to this agent.

\begin{lemma}\label{if alpha then s}
If $\alpha_i\in A^n_s$, then $s(\alpha_i)\ge p_i$, for each $A\subseteq\mathcal{A}_0$.
\end{lemma}
\begin{proof}
Suppose that $\alpha_i\in A^n_s$. Note that $\alpha_i\notin \mathcal{A}_0\supseteq A =A^0_s$ by the choice of $\alpha_1,\dots,\alpha_m$. Thus, by Definition~\ref{Ak}, there is $k< n$ such that
$$
\lambda(\alpha_i)\cdot\ s(\alpha_i) + \sum_{a\in A^{k}_s} w(a,\alpha_i) \ge \theta(\alpha_i).
$$
By Definition~\ref{canonical w}, $w(a,\alpha_i)=0$ for each $a\in\mathcal{A}$. Hence,
$
\lambda(\alpha_i)\cdot\ s(\alpha_i) \ge \theta(\alpha_i).
$
By Definition~\ref{canonical lambda}, $\lambda(\alpha_i)=1$. By Definition~\ref{canonical theta}, $\theta(\alpha_i)=p_i$. Therefore, $s(\alpha_i)\ge p_i$.
\end{proof}

\begin{lemma}\label{beta step lemma}
For each  $n\ge 0$ and each subset $A$ of $\mathcal{A}_0$,
if $\beta_i\in A^{n+1}_s\setminus A^{n}_s$, then
$X\vdash (A^{n}_s\cap \mathcal{A}_0)\rhd_{p_i} B_i$.
\end{lemma}
\begin{proof}
By Definition~\ref{Ak}, assumption $\beta_i\in A^{n+1}_s\setminus A^{n}_s$ implies that
$$
\lambda(\beta_i)\cdot s(\beta_i) + \sum_{a\in A^{n}_s}w(a,\beta_i)\ge \theta(\beta_i).
$$
By Definition~\ref{canonical lambda}, $\lambda(\beta_i)=0$. Thus,
$
\sum_{a\in A^{n}_s}w(a,\beta_i)\ge \theta(\beta_i).
$
Hence, by Definition~\ref{canonical theta},
$
\sum_{a\in A^{n}_s}w(a,\beta_i)\ge |A_i|+1
$.
Thus, $A_i\subseteq A^n_s$ by Lemma~\ref{canonical w lemma}.
Recall that $A_i\subseteq \mathcal{A}_0$ by the choice of set $X$. Hence, $A_i\subseteq A^n_s\cap \mathcal{A}_0$. Then, $\vdash (A^n_s\cap \mathcal{A}_0)\rhd_0 A_i$ by Reflexivity axiom. Recall that $A_i\rhd_{p_i} B_i\in X$. Thus,  $X \vdash (A^n_s\cap \mathcal{A}_0)\rhd_{p_i} B_i$ by Transitivity axiom.
\end{proof}

\begin{lemma}\label{many beta steps lemma}
$X\vdash (A^{n}_s\cap \mathcal{A}_0)\rhd_{q}\bigcup_{\beta_i\in A^{n+1}_s\setminus A^{n}_s} B_i$, where $q=\sum_{\beta_i\in A^{n+1}_s\setminus A^{n}_s}p_i$.
\end{lemma}
\begin{proof}
The statement of the lemma follows from Lemma~\ref{beta step lemma} and Lemma~\ref{big union lemma}.
\end{proof}

\begin{lemma}\label{n>0 lemma}
$X\vdash (A^{n}_s\cap \mathcal{A}_0)\rhd_{q}((A^{n+2}_s\setminus A^{n+1}_s)\cap \mathcal{A}_0)$,
where $q=\sum_{\beta_i\in A^{n+1}_s\setminus A^{n}_s}p_i$, for each subset $A$ of $\mathcal{A}_0$, each spending function $s$, and each $n\ge 0$.
\end{lemma}
\begin{proof}
By Definition~\ref{Ak},
\begin{eqnarray*}
(A^{n+2}_s\setminus A^{n+1}_s)\cap \mathcal{A}_0  &=& \left\{b\in \mathcal{A}_0\;\middle|\; \lambda(b)\cdot s(b) + \sum_{a\in A^{n+1}_s}w(a,b)\ge \theta(b)\right\}\\
&\setminus&
\left\{b\in \mathcal{A}_0\;\middle|\; \lambda(b)\cdot s(b) + \sum_{a\in A^{n}_s}w(a,b)\ge \theta(b)\right\}.
\end{eqnarray*}
By Definition~\ref{canonical lambda}, $\lambda(b)=0$ for all $b\in\mathcal{A}_0$. By Definition~\ref{canonical theta}, $\theta(b)=1$ for all $b\in\mathcal{A}_0$.  Thus,
\begin{eqnarray*}
(A^{n+2}_s\setminus A^{n+1}_s)\cap \mathcal{A}_0  = \left\{b\in \mathcal{A}_0\;\middle|\;  \sum_{a\in A^{n+1}_s}w(a,b)\ge 1\right\} \setminus \left\{b\in \mathcal{A}_0\;\middle|\;  \sum_{a\in A^{n}_s}w(a,b)\ge 1\right\}.
\end{eqnarray*}
Since $\alpha_1,\dots,\alpha_m,\beta_1,\dots\beta_m\notin \mathcal{A}_0$, by Definition~\ref{canonical w}, for each $b\in\mathcal{A}_0$, we have  $w(a,b)\neq 0$ only if $a=\beta_i$ and $b\in B_i$ for some $i\le m$, in which case $w(a,b) = 1$. Hence,
\begin{eqnarray*}
(A^{n+2}_s\setminus A^{n+1}_s)\cap \mathcal{A}_0  = \bigcup_{\beta_i\in A^{n+1}_s\setminus A^{n}_s} B_i.
\end{eqnarray*}
Thus, to finish the proof of the lemma, it is sufficient to show that
$$X\vdash (A^{n}_s\cap \mathcal{A}_0)\rhd_{q}\bigcup_{\beta_i\in A^{n+1}_s\setminus A^{n}_s} B_i,$$
where $q=\sum_{\beta_i\in A^{n+1}_s\setminus A^{n}_s}p_i$, which follows from Lemma~\ref{many beta steps lemma}.
\end{proof}

\begin{lemma}\label{n=0 lemma}
$X\vdash A\rhd_{q}(A^{1}_s\cap \mathcal{A}_0)$, for each subset $A$ of $\mathcal{A}_0$, each spending function $s$, and each non-negative real number $q$.
\end{lemma}
\begin{proof}
By Definition~\ref{Ak},
\begin{eqnarray*}
A^{1}_s\cap \mathcal{A}_0  &=& (A \cap \mathcal{A}_0) \cup \left\{b\in \mathcal{A}_0 \;\middle|\; \lambda(b)\cdot s(b) + \sum_{a\in A}w(a,b)\ge \theta(b)\right\}.
\end{eqnarray*}
By Definition~\ref{canonical lambda}, $\lambda(b)=0$ for all $b\in\mathcal{A}_0$. By Definition~\ref{canonical theta}, $\theta(b)=1$ for all $b\in\mathcal{A}_0$.  Thus,
\begin{eqnarray*}
A^{1}_s\cap \mathcal{A}_0  &=& (A \cap \mathcal{A}_0) \cup \left\{b\in \mathcal{A}_0 \;\middle|\; \sum_{a\in A}w(a,b)\ge 1\right\}.
\end{eqnarray*}
By Definition~\ref{canonical w}, $w(a,b)=0$  for all $a\in A\subseteq \mathcal{A}_0$ and all $b\in \mathcal{A}_0$. Thus, the set 
$$
\left\{b\in \mathcal{A}_0 \;\middle|\; \sum_{a\in A}w(a,b)\ge 1\right\}
$$
is empty. Hence,
$A^{1}_s\cap \mathcal{A}_0 = A \cap \mathcal{A}_0$. Therefore, $X\vdash A\rhd_{q}(A^{1}_s\cap \mathcal{A}_0)$ by Reflexivity axiom.
\end{proof}

\begin{lemma}\label{An+1 lemma}
$X\vdash A\rhd_{q} (A^{n+1}_s\cap \mathcal{A}_0)$, where $q=\sum_{\beta_i\in A^n_s}p_i$, for each subset $A$ of $\mathcal{A}_0$, each spending function $s$, and each $n\ge 0$.
\end{lemma}
\begin{proof}
We prove this lemma by induction on $n$.
If $n=0$, then the required follows from Lemma~\ref{n=0 lemma}. 

Assume now that
\begin{equation}\label{eq-alpha}
X\vdash A\rhd_{q} (A^{n+1}_s\cap \mathcal{A}_0),
\end{equation}
where $q=\sum_{\beta_i\in A^n_s}p_i$.

Note that $A^n_s\subseteq A^{n+1}_s$ by Definition~\ref{Ak}. Hence, $\vdash (A^{n+1}_s\cap \mathcal{A}_0) \rhd_0 (A^{n}_s\cap \mathcal{A}_0)$ by Reflexivity axiom. 
At the same time, by Lemma~\ref{n>0 lemma},
$$X\vdash (A^{n}_s\cap \mathcal{A}_0)\rhd_{r}((A^{n+2}_s\setminus A^{n+1}_s)\cap \mathcal{A}_0),$$
 where $r=\sum_{\beta_i\in A^{n+1}_s\setminus A^{n}_s}p_i$. 
 Thus, by Transitivity axiom,
 $$
 X\vdash (A^{n+1}_s\cap \mathcal{A}_0)\rhd_{r}((A^{n+2}_s\setminus A^{n+1}_s)\cap \mathcal{A}_0).
$$
 Then, by Augmentation axiom,
$$
 X\vdash (A^{n+1}_s\cap \mathcal{A}_0),(A^{n+1}_s\cap \mathcal{A}_0)\rhd_{r}((A^{n+2}_s\setminus A^{n+1}_s)\cap \mathcal{A}_0),(A^{n+1}_s\cap \mathcal{A}_0).
$$
In other words,
 \begin{equation}\label{eq-beta}
 X\vdash (A^{n+1}_s\cap \mathcal{A}_0)\rhd_{r}(A^{n+2}_s\cap \mathcal{A}_0).
 \end{equation}
Therefore, by Transitivity  axiom from statement (\ref{eq-alpha}) and (\ref{eq-beta}) we can conclude 
$X\vdash A\rhd_{q'} (A^{n+2}_s\cap \mathcal{A}_0)$, where $q'=q+r=\sum_{\beta_i\in A^{n+1}_s}p_i$.
\end{proof}

We are now ready to prove the converse of Lemma~\ref{first base case}.
\begin{lemma}\label{second base case}
If $N_X\vDash A\rhd_{p} B$, then $X\vdash A\rhd_{p} B$, for each subsets $A$ and $B$ of $\mathcal{A}_0$ and each non-negative real number $p$.
\end{lemma}
\begin{proof}
Suppose that $N_X\vDash A\rhd_{p} B$. By Definition~\ref{sat}, there is a spending function $s$ such that $\|s\|\le p$ and $B\subseteq A^*_s$. Thus, by Corollary~\ref{A*=Ak}, there is $n\ge 0$ such that $B\subseteq A^n_s$. By Definition~\ref{Ak}, $A^{n}_s\subseteq A^{n+1}_s$. Thus, $B\subseteq A^{n+1}_s$.
Since $B$ is a subset of $\mathcal{A}_0$, we have $B\subseteq A^{n+1}_s\cap \mathcal{A}_0$.
Hence, $\vdash (A^{n+1}\cap \mathcal{A}_0)\rhd_0 B$ by Reflexivity axiom. Then, from Transitivity axiom and Lemma~\ref{An+1 lemma}, we have $X\vdash A\rhd_{q} B$, where $q=\sum_{\beta_i\in A^n_s}p_i$.

Note that $\sum_{\beta_i\in A^n_s}p_i \le \sum_{\alpha_i\in A^{n}_s}p_i$ by Lemma~\ref{if beta then alpha} and $\sum_{\alpha_i\in A^{n}_s}p_i \le \sum_{s(\alpha_i)\ge p_i}p_i$ by Lemma~\ref{if alpha then s}. Thus, taking into account Definition~\ref{||set},
$$
q=\sum_{\beta_i\in A^n_s}p_i \le \sum_{\alpha_i\in A^{n}_s}p_i \le \sum_{s(\alpha_i)\ge p_i}p_i \le \sum_{s(\alpha_i)\ge p_i}s(\alpha_i) \le \sum_{a\in\mathcal{A}}s(a)=\|s\|\le p.
$$
Hence, $q\le p$. Then, $\vdash B\rhd_{p-q} B$ by Reflexivity axiom. Finally, $X\vdash A\rhd_{q} B$ and $\vdash B\rhd_{p-q} B$, by Transitivity axiom, imply that $X\vdash A\rhd_{p} B$.
\end{proof}

We conclude this section by stating and proving the completeness theorem for promotional marketing.
\begin{theorem}\label{promotional completeness}
If $\nvdash\phi$, then there exists social network $N=(\mathcal{A},w,\lambda,\theta)$ such that $\phi\in\Phi(\mathcal{A})$ and $N\nvDash\phi$.
\end{theorem}
\begin{proof}
Suppose that $\nvdash\phi$. Let $M$ be any maximal consistent subset of 
$$\{\psi,\neg\psi\;|\; \mbox{$\psi$ is a subformula of $\neg\psi$} \}$$ 
such that $\neg\psi\in M$. Let $X$ be the set of all atomic formulas of the form $A\rhd_p B$ in set $M$. To finish the proof of the theorem, we first establish the following lemma:

\begin{lemma}\label{induction lemma}
$\psi\in M$ if and only if $N_X\vDash\psi$ for each subformula $\psi$ of $\neg\phi$.
\end{lemma}
\begin{proof}
We prove the lemma by induction on the structural complexity of formula $\psi$. In the base case, suppose that $\psi$ is $A\rhd_p B$. 

\noindent $(\Rightarrow)$ If $A\rhd_p B\in M$, then $A\rhd_p B\in X$ by the choice of set $X$. Thus, $N_X\vDash A\rhd_p B$ by Lemma~\ref{first base case}.

\noindent $(\Leftarrow)$ If $N_X\vDash A\rhd_p B$, then $X\vdash A\rhd_p B$ by Lemma~\ref{second base case}. Thus, $M\vdash A\rhd_p B$. Hence, by the maximality of set $M$, we have $A\rhd_p B\in M$ since $A\rhd_p B$ is a subformula of $\neg\phi$. 

The induction step follows from the induction hypothesis, the maximality and the consistency of set $M$ and Definition~\ref{sat} in the standard way.
\end{proof}

To finish the proof of the theorem, note that $\neg\phi\in M$ by the choice of set $M$. Thus, $N_X\vDash\neg\psi$ by Lemma~\ref{induction lemma}. Therefore, $N_X\nvDash\psi$ by Definition~\ref{sat}.
\end{proof}

\section{Logic of Preventive Marketing}\label{preventive section}

In this section we study the impact of preventive marketing on influence in social networks. Our definition of a social network given in Definition~\ref{social network definition} and the language $\Phi(\mathcal{A})$ remain the same. As it has been discussed in the introduction, we only modify the meaning of the influence relation $A\rhd_p B$ to be ``for any preventive marketing campaign with budget no more than $p$, the set of agents $A$ is able to influence the set of agents $B$". The latter is formally captured in item 1 of Definition~\ref{preventive sat}.

\begin{definition}\label{preventive sat}
For any social network $N$ with the set of agents $\mathcal{A}$ and any formula $\phi\in\Phi(\mathcal{A})$, we define the satisfiability relation $N\vDash\phi$ as follows:
\begin{enumerate}
\item $N\vDash A \rhddash_p B$ if $B\subseteq A^*_{s}$ for each spending function $s$ such that $\|s\|\le p$,
\item $N\vDash\neg\psi$ if $N\nvDash\psi$,
\item $N\vDash\psi\to\chi$ if $N\nvDash\psi$ or $N\vDash\chi$.
\end{enumerate}
\end{definition}

Note the significant difference between the above definition and the similar Definition~\ref{sat} for promotional marketing. Item 1 of Definition~\ref{preventive sat} has a universal quantifier over spending functions and corresponding part of Definition~\ref{sat} has an existential quantifier over spending functions.

\subsection{Axioms}\label{preventive axioms}
Let $\mathcal{A}$ be any fixed finite set of agents. Our logical system for influence with preventive marketing, in addition to propositional tautologies in language $\Phi(\mathcal{A})$, contains the following axioms:
\begin{enumerate}
\item Reflexivity: $A\rhddash_p B$, where $B\subseteq A$,
\item Augmentation: $A\rhddash_p B\to A,C\rhddash_p B,C$,
\item Transitivity: $A\rhddash_p B\to(B\rhddash_p C\to A\rhddash_p C)$,
\item Monotonicity: $A\rhddash_p B\to A\rhddash_q B$, where $q\le p$.
\end{enumerate}
Just like in the case of promotional marketing, we write $\vdash \phi$ if formula $\phi\in\Phi(\mathcal{A})$ is derivable in our logical system using Modus Ponens inference rule. We write $X\vdash \phi$ if formula $\phi$ is derivable using an additional set of axioms $X\subseteq \Phi(\mathcal{A})$.

\subsection{Example}
The soundness and the completeness of our logical system will be shown later. In this section we give two examples of formal proofs in our system. First, we show a preventive marketing analogy of Lemma~\ref{union lemma}:

\begin{lemma}\label{preventive union lemma}
$\vdash A\rhddash_p B\to (A\rhddash_p C\to A\rhddash_p B,C)$. 
\end{lemma}
\begin{proof}
By Augmentation axiom,
\begin{equation}\label{prev eq 1}
\vdash A\rhd_p B \to A\rhd_p A,B
\end{equation}
and 
\begin{equation}\label{prev eq 2}
\vdash A\rhd_p C \to A,B\rhd_p B,C.
\end{equation}
By Transitivity axiom,
\begin{equation}\label{prev eq 3}
\vdash A\rhd_p A,B \to (A,B\rhd_p B,C \to A\rhd_p B,C).
\end{equation}
The statement of the lemma follows from statements (\ref{prev eq 1}), (\ref{prev eq 2}), and (\ref{prev eq 3}) by the laws of the propositional logic.
\end{proof}

Next, we show an auxiliary lemma that is used later in the proof of completeness.
\begin{lemma}\label{rhd set}
If $X\vdash B\rhddash_p c$ for each $c\in C$, then $X\vdash B\rhddash_p C$, where $B$ and $C$ are subsets of $\mathcal{A}_0$ and $p\ge 0$.
\end{lemma}
\begin{proof}
We prove the lemma by induction on the size of set $C$. 

\noindent {\em Base Case}: $X\vdash B\rhddash_p \varnothing$ by Reflexivity axiom.

\noindent {\em Induction Step}: 
Assume that  $X\vdash B\rhddash_p C$. Let $c$ be any element of $\mathcal{A}_0\setminus C$ such that $X\vdash B\rhddash_p c$. We need to show that $X\vdash B\rhddash_p C\cup \{c\}$. By Augmentation axiom, 
\begin{equation}\label{eq0}
X\vdash B\cup \{c\}\rhddash_p C\cup \{c\}.
\end{equation}
Recall that $X\vdash B\rhddash_p c$. Again by Augmentation axiom, $X\vdash B\rhddash_p B \cup \{c\}$. Hence, $X\vdash B \rhddash_p C\cup \{c\}$, due to (\ref{eq0}) and Transitivity axiom.
\end{proof}

\subsection{Soundness}

In this section we prove the soundness of the logic for preventive marketing. 

\begin{theorem}\label{preventive soundness}
For any finite set $\mathcal{A}$ and any $\phi\in\Phi(\mathcal{A})$,
if $\vdash \phi$, then $N\vDash\phi$ for each social network $N=(\mathcal{A},w,\lambda,\theta)$. \end{theorem}

The soundness of propositional tautologies and of Modus Ponens inference rule is straightforward. Below we show the soundness of each of the remaining axioms as a separate lemma.

\begin{lemma}
$N\vDash A\rhd_p B$, for any social network $N=(\mathcal{A}, w, \lambda, \theta)$ and any subsets $A$ and $B$ of $\mathcal{A}$ such that $B\subseteq A$.
\end{lemma}
\begin{proof} Let $s$ be any spending function. By Definition~\ref{preventive sat}, it suffices to show that $B\subseteq A^*_s$. Indeed,  $A\subseteq A^*_s$ by Corollary~\ref{A subseteq A*s}. Therefore, $B\subseteq A^*_s$ due to the assumption $B\subseteq A$ of the lemma.
\end{proof}

\begin{lemma}
If $N\vDash A\rhd_p B$, then $N\vDash A,C\rhd_p B,C$, for each social network $N=(\mathcal{A}, w, \lambda, \theta)$ and all subsets $A$, $B$, and $C$ of $\mathcal{A}$.
\end{lemma}
\begin{proof}
Suppose that $N \vDash A \rhd_p B $. Consider any spending function $s$ such that $\|s\|\le p$. It suffices to show that 
$B\cup C \subseteq (A\cup C)^*_s$. Indeed, assumption $N \vDash A \rhd_p B $ implies that $B\subseteq A^*_s$ by Definition~\ref{preventive sat}. At the same time, $C\subseteq C^*_s$ by Corollary~\ref{A subseteq A*s}.  Therefore, 
$
B\cup C \subseteq A^*_s \cup C^*_s \subseteq (A\cup C)^*_s
$, by Lemma~\ref{star union}.
\end{proof}

\begin{lemma} 
If $N \vDash A \rhd_p B$ and $N\vDash B \rhd_p C$, then $N \vDash A \rhd_{p} C$, for each social network $N=(\mathcal{A},w,\lambda,\theta)$ and all subsets $A$, $B$, and $C$ of $\mathcal{A}$.
\end{lemma}
\begin{proof}
Suppose that $N \vDash A \rhd_p B$ and $N\vDash B \rhd_p C$. Consider any spending function $s$ such that $\|s\|\le p$. By Definition~\ref{preventive sat}, it suffices to show that $C\subseteq A^*_s$. 

Note that assumption $N\vDash A \rhd_p B$, by Definition~\ref{preventive sat}, imply that $B\subseteq A^*_s$. Thus, $B^*_s\subseteq (A^*_s)^*_s$ by Corollary~\ref{star mono}. At the same time, assumption $N\vDash B \rhd_p C$ implies that $C\subseteq B^*_s$ by Definition~\ref{preventive sat}. Hence, $C\subseteq (A^*_s)^*_s$. Therefore, $C\subseteq A^*_s$ by Lemma~\ref{idempotent}.
\end{proof}

\begin{lemma}
If $N\vDash A\rhd_p B$, then  $N\vDash A\rhd_q B$, for each $q\le p$, each each social network $N=(\mathcal{A},w,\lambda,\theta)$ and all subsets $A$ and $B$ of $\mathcal{A}$.
\end{lemma}
\begin{proof}
Consider any spending function $s$ such that $\|s\|\le q$. By Definition~\ref{preventive sat}, it suffices to show that $B\subseteq A^*_s$. To prove this, note that $\|s\|\le q\le p$. Thus, $B\subseteq A^*_s$ due to Definition~\ref{preventive sat} and the assumption $N\vDash A\rhd_p B$ of the lemma.
\end{proof}

This concludes the proof of the soundness of our logical system for preventive marketing.

\subsection{Completeness}

The rest of this section contains the proof of the following result.
\begin{theorem}\label{preventive completeness}
If $\nvdash\phi$, then there is a social network $N=(\mathcal{A},w,\lambda,\theta)$ such that $\phi\in\Phi(\mathcal{A})$ and $N\nvDash\phi$.
\end{theorem}

Suppose that $\nvdash\phi$. It suffices to construct a ``canonical" social network $N=(\mathcal{A},w,\lambda,\theta)$ such that $N\nvDash\phi$. Define $P\subset \mathbb{R}$ to be the finite set of all subscripts that appear in formula $\phi$. Let $\epsilon>0$ be such that $|p_1-p_2|>\epsilon$ for all $p_1,p_2\in P$ where  $p_1\neq p_2$. Let $\mathcal{A}_0$ be the finite set of all agents that appear in formula $\phi$ and $X$ be a maximal consistent subset of $\Phi(\mathcal{A}_0)$ containing formula $\neg\phi$.

In Section~\ref{social networks section}, we have introduced closures $A^k_s$ and $A^*_s$ of a set of agents $A$. Both of these closures are {\em semantic} in the sense that they are defined in terms of a given social network. We are about to introduce another closure $A^+_p$ that will be used to construct the canonical social network $N$. Unlike closures $A^k_s$ and $A^*_s$, closure the $A^+_p$ is {\em syntactic} because it is defined in terms of provability of certain statements in our logical system.

\begin{definition}\label{A+}
$A^{+}_p=\{a\in\mathcal{A}_0\;|\; X\vdash A\rhddash_p  a\}$, for any set of agents $A\subseteq\mathcal{A}_0$ and any $p\ge 0$.
\end{definition}

\begin{lemma}\label{ArhdA+}
$X \vdash A\rhddash_p A^+_p$, for any $A\subseteq\mathcal{A}_0$ and any $p\ge 0$.
\end{lemma}
\begin{proof}
The statement of the lemma follows from Definition~\ref{A+} and Lemma~\ref{rhd set}.
\end{proof}

Generally speaking, it is possible that $A^+_p=A^+_q$ for some $p$ and $q$ such that $p\neq q$. In the construction of the canonical social network $N$ it will be convenient to distinguish closures $A^+_p$ for different values of parameter $p$. In such situations, instead of closure $A^+_p$ we consider labeled closure, formally defined as pair $(A^+_p,p)$.

\begin{definition}\label{Z}
Let $\mathbb{L}=\{(A^+_p,p)\;|\; A\subseteq \mathcal{A}_0, p\in P\}$.
\end{definition}

Next we define the canonical network $N=(\mathcal{A},w,\lambda,\theta)$. Besides agents in set $\mathcal{A}_0$, our social network also has two additional agents for each $\ell\in \mathbb{L}$. By analogy with the canonical social network $N_X$ from the proof of completeness for promotional marking, we call these additional agents $\alpha(\ell)$ and $\beta(\ell)$.

\begin{definition}\label{A}
$\mathcal{A}=\mathcal{A}_0\cup \{\alpha(\ell),\beta(\ell)\;|\; \ell\in \mathbb{L}\}$.
\end{definition}

\begin{figure}[ht]
\begin{center}
\vspace{0mm}
\scalebox{.6}{\includegraphics{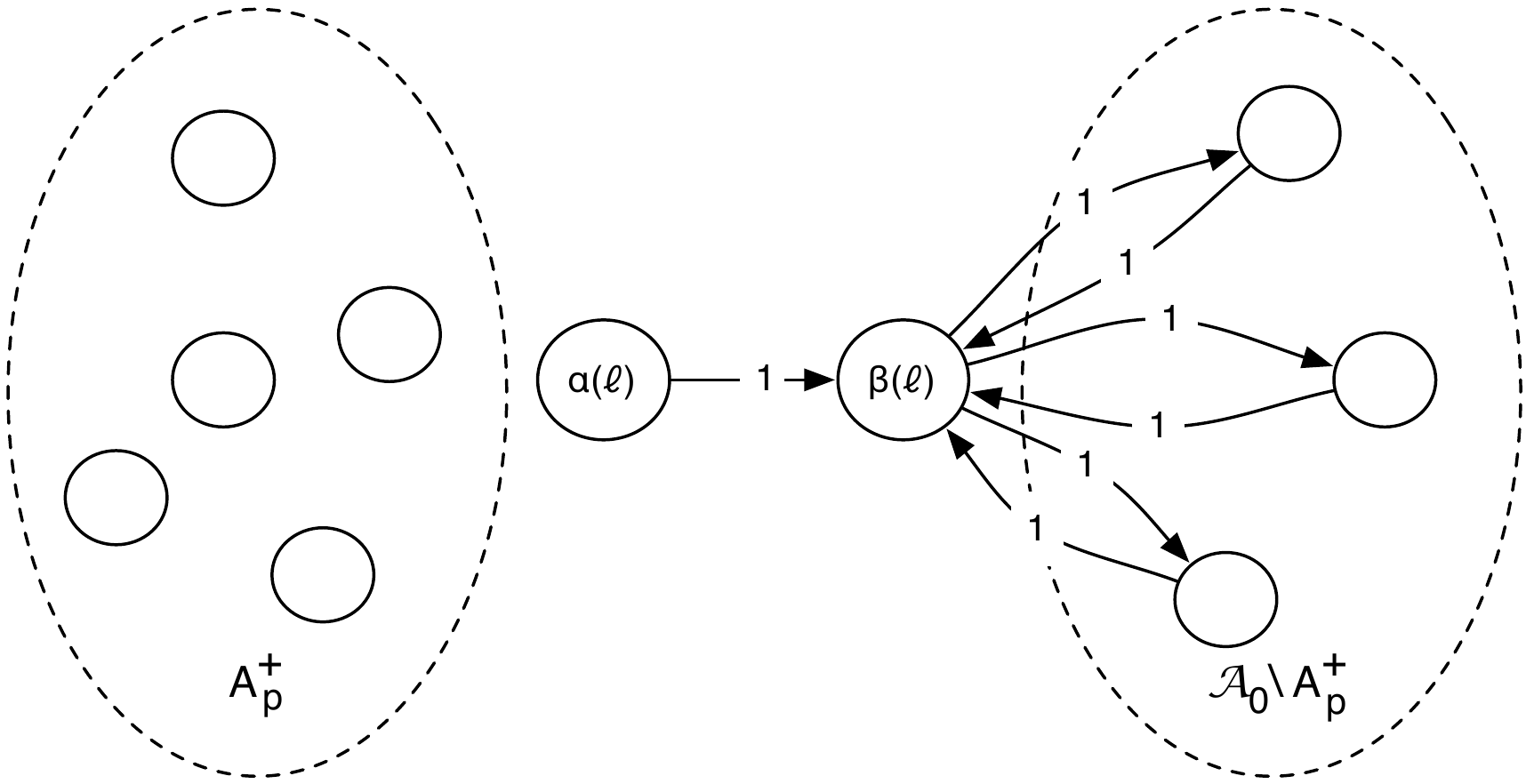}}
\vspace{0mm}
\footnotesize
\caption{Towards the definition of the influence function in the canonical social network.}\label{second canonical network figure}
\vspace{0cm}
\end{center}
\vspace{0mm}
\end{figure}

For any $\ell=(A^+_p,p)$, we assume that upon adopting the product agent $\alpha(\ell)$ puts peer pressure on agent $\beta(\ell)$, agent $\beta(\ell)$ puts peer pressure on each agent in set $\mathcal{A}_0\setminus A^+_p$ and each agent in set $\mathcal{A}_0\setminus A^+_p$, in turn, puts peer pressure on agent $\beta(\ell)$. The peer pressure structure is illustrated in Figure~\ref{second canonical network figure}. Note that the same agent $a\in \mathcal{A}_0$ can belong to set $\mathcal{A}_0\setminus A^+_p$ for several different values of $p$. Such agent $a$ could experience (or put) peer pressure from (on) several different agents $\beta(\ell)$. The structure is formally specified in Definition~\ref{canonical w minus}.

\begin{definition}\label{canonical w minus}
For any $a,b\in \mathcal{A}$,
$$
w(a,b)=
\begin{cases}
1, & \mbox{if $a=\alpha(\ell)$ and $b=\beta(\ell)$ for some $\ell\in \mathbb{L}$},\\
1, & \mbox{if $a\in \mathcal{A}_0\setminus A^+_p$ and $b=\beta(A^+_p,p)\in \mathbb{L}$},\\
1, & \mbox{if $a=\beta(A^+_p,p)\in \mathbb{L}$ and $b\in \mathcal{A}_0\setminus A^+_p$},\\
0, & \mbox{otherwise}.
\end{cases}
$$
\end{definition}

We assume that only agents $\{\alpha(\ell)\;|\;\ell\in \mathbb{L}\}$ are responsive to preventive marketing. This is formally captured in the definition of the propensity function below.

\begin{definition}\label{canonical lambda minus}
For any $a\in \mathcal{A}$,
$$
\lambda(a)=
\begin{cases}
-1, & \mbox{if $a=\alpha(\ell)$ for some $\ell\in \mathbb{L}$},\\
0, & \mbox{otherwise}.
\end{cases}
$$
\end{definition}

To finish the definition of canonical social network $N=(\mathcal{A},w,\lambda,\theta)$, we only need to define threshold function $\theta(a)$ for each $a\in\mathcal{A}$. There are three different cases to consider: $a=\alpha(\ell)$ for some $\ell\in \mathbb{L}$, $a=\beta(\ell)$ for some $\ell\in \mathbb{L}$, and $a\in\mathcal{A}_0$.

Recall that by Definition~\ref{canonical w minus} and Definition~\ref{canonical lambda minus}, agent $\alpha(\ell)$ is not responsive to peer pressure of any other agent. It is only responsive to the  marketing pressure with propensity $-1$. We set the threshold value of this agent to $\epsilon-p$, where $\ell=(A^+_p,p)$. Thus, if an amount at least $p$ is spent on the preventive marketing to this agent, it will not adopt the product.

We set threshold value of agent $\beta(\ell)$ to 1. Thus, for each $\ell=(A^+_p,p)$,  if either agent $\alpha(\ell)$ or any of the agents in the set  $\mathcal{A}_0\setminus A^+_p$ adopts the product, then agent $\beta(\ell)$ will also adopt the product.

Finally, recall from Definition~\ref{canonical w minus} that agent $a\in\mathcal{A}_0$ can experience peer pressure from any agent $\beta(A^+_p,p)$ such that $a\in \mathcal{A}_0\setminus A^+_p$. There are exactly $|\{(A^+_p,p)\in \mathbb{L}\;|\; a\in \mathcal{A}_0\setminus A^+_p\}|$ such $\beta$-agents. We set the threshold value $\theta(a)$ high enough so that it adopts the product only if {\em all} of these $\beta$-agents adopt the product.

The next definition captures the three cases discussed above.

\begin{definition}\label{canonical theta minus}
For any $a\in \mathcal{A}$,
$$
\theta(a)=
\begin{cases}
\epsilon-p, & \mbox{if $a=\alpha(A^+_p,p)$},\\
1, & \mbox{if $a=\beta(A^+_p,p)$},\\
|\{(A^+_p,p)\in \mathbb{L}\;|\; a\in \mathcal{A}_0\setminus A^+_p\}|, & \mbox{if $a\in\mathcal{A}_0$}.
\end{cases}
$$
\end{definition}

For any $c\in\mathcal{A}_0$, we have chosen $\theta(c)$ to be equal to the number of $\beta(A^+_p,p)$ such that $c\in\mathcal{A}_0\setminus A^+_p$. Thus, if all such $\beta$-agents adopt the product, then the total peer pressure on agent $c$ would reach $\theta(c)$ and agent $c$ also would adopt the product. This observation is formalized by the next lemma.

\begin{lemma}\label{or lemma}
Let $c$ be an agent in $\mathcal{A}_0$, set $B$ be a subset of $\mathcal{A}_0$, and $s$ be an arbitrary spending function for the social network $N$. If for each $(A^+_p,p)\in \mathbb{L}$ at least one of the following is true: (i) $c\in A^+_p$, (ii) $\beta(A^+_p,p)\in B^*_{s}$, then $c\in B^*_{s}$.
\end{lemma}
\begin{proof}
By Corollary~\ref{A*=Ak}, $B^*_{s}=B^n_{s}$ for some $n\ge 0$. Thus, by the assumption of this lemma, for each $(A^+_p,p)\in \mathbb{L}$ at least one of the following is true: (i) $c\in A^+_p$, (ii) $\beta(A^+_p,p)\in B^n_{s}$. In other words, 
$\{\beta(A^+_p,p)\;|\; c\in \mathcal{A}_0\setminus A^+_p\}\subseteq B^n_{s}$. Hence,
$$
\sum_{b\in B^n_{s}}w(b,c) \ge \sum_{\ell\in \{(A^+_p,p)\in \mathbb{L}\;|\; c\in \mathcal{A}_0\setminus A^+_p\}} w(\beta(\ell),c).
$$
Thus, by Definition~\ref{canonical w minus},
$$
\sum_{b\in B^n_{s}}w(b,c) \ge \sum_{\ell\in \{(A^+_p,p)\in \mathbb{L}\;|\; c\in \mathcal{A}_0\setminus A^+_p\}} 1=
|\{(A^+_p,p)\in \mathbb{L}\;|\; c\in \mathcal{A}_0\setminus A^+_p\}|.
$$
At the same time $\lambda(c)=0$ by Definition~\ref{canonical lambda minus}. Hence,
$$
\lambda(c)\cdot s(c) + \sum_{b\in B^n_{s}}w(b,c) \ge |\{(A^+_p,p)\in \mathbb{L}\;|\; c\in \mathcal{A}_0\setminus A^+_p\}|.
$$
Then, by Definition~\ref{canonical theta minus},
$$
\lambda(c)\cdot s(c) + \sum_{b\in B^n_{s}}w(b,c) \ge \theta(c).
$$
Thus, $c\in B^{n+1}_{s}$ by Definition~\ref{Ak}. Therefore, $c\in B^*_{s}$ by Definition~\ref{A*}.
\end{proof}

Referring back to Figure~\ref{second canonical network figure}, note that if an agent in set $\mathcal{A}_0\setminus A^+_p$ adopts the product, then it will put enough pressure on $\beta(A^+_p,p)$ so that agent $\beta(A^+_p,p)$ also adopts the product. We formally state this observation as the lemma below.

\begin{lemma}\label{if rhs then beta}
If there is $b_0\in B^*_{s}$ such that $b_0\in\mathcal{A}_0\setminus A^+_p$, then $\beta(A^+_p,p)\in B^*_{s}$, where $(A^+_p,p)\in \mathbb{L}$, set $B$ is a subset of $\mathcal{A}_0$, and $s$ is an arbitrary spending function for the social network $N$.
\end{lemma}
\begin{proof}
By Corollary~\ref{A*=Ak}, $B^*_{s}=B^n_{s}$ for some $n\ge 0$. At the same time, by Definition~\ref{canonical w minus}, assumption $b_0\in\mathcal{A}_0\setminus A^+_p$ implies that $w(b_0,\beta(A^+_p,p))=1$. Thus, 
$$
\sum_{b\in B^n_{s}}w(b,\beta(A^+_p,p))\ge w(b_0,\beta(A^+_p,p))=1,
$$
since $b_0\in B^*_{s}=B^n_{s}$.
Note that $\lambda(\beta(A^+_p,p))=0$ by Definition~\ref{canonical lambda minus}. Hence,
$$
\lambda(\beta(A^+_p,p))\cdot s(\beta(A^+_p,p)) + \sum_{b\in B^n_{s}}w(b,\beta(A^+_p,p))\ge w(b_0,\beta(A^+_p,p))=1.
$$
Thus, by Definition~\ref{canonical theta minus},
$$
\lambda(\beta(A^+_p,p))\cdot s(\beta(A^+_p,p)) + \sum_{b\in B^n_{s}}w(b,\beta(A^+_p,p))\ge \theta(\beta(A^+_p,p)).
$$
Hence, $\beta(A^+_p,p)\in B^{n+1}_{s}$ by Definition~\ref{Ak}. Therefore, $\beta(A^+_p,p)\in B^*_{s}$ by Definition~\ref{A*}.
\end{proof}

Recall that we have set the threshold value of agent $\alpha(A^+_p,p)$ to be $\epsilon - p$, so that by spending at least $p$ on preventive marketing one would prevent an adoption of the product by agent $\alpha(A^+_p,p)$. At the same time, spending no more than $p-\epsilon$ will result in $\alpha(A^+_p,p)$ adopting the product. Once agent $\alpha(A^+_p,p)$ adopts the product, it will put enough pressure on agent $\beta(A^+_p,p)$ to adopt the product as well. This observation is formally stated below. 

\begin{lemma}\label{if alpha then beta}
If $s(\alpha(A^+_p,p))\le p-\epsilon$, then $\beta(A^+_p,p)\in B^*_{s}$,
where $(A^+_p,p)\in \mathbb{L}$, set $B$ is a subset of $\mathcal{A}_0$, and $s$ is an arbitrary spending function for the social network $N$.
\end{lemma}
\begin{proof}
Suppose that $s(\alpha(A^+_p,p))\le p-\epsilon$. Note that $\lambda(\alpha(A^+_p,p))=-1$ by Definition~\ref{canonical lambda minus} and $w(b,\alpha(A^+_p,p))=0$ for each $b\in\mathcal{A}$ by Definition~\ref{canonical w minus}. Thus,
\begin{eqnarray*}
&&\lambda(\alpha(A^+_p,p))\cdot s(\alpha(A^+_p,p)) + \sum_{b\in B^0_{s}}w(b,\alpha(A^+_p,p))\\
&&=-1\cdot s(\alpha(A^+_p,p)) + 0 = -s(\alpha(A^+_p,p)) \ge \epsilon - p.
\end{eqnarray*}
Thus, by Definition~\ref{canonical theta minus},
$$
\lambda(\alpha(A^+_p,p))\cdot s(\alpha(A^+_p,p)) + \sum_{b\in B^0_{s}}w(b,\alpha(A^+_p,p)) \ge \theta(\alpha(A^+_p,p)).
$$
Hence, $\alpha(A^+_p,p)\in B^1_{s}$ by Definition~\ref{Ak}. Since $w$ is a non-negative function, by Definition~\ref{canonical w minus},
\begin{eqnarray*}
\sum_{b\in B^1_{s}}w(b,\beta(A^+_p,p))&=& w(\alpha(A^+_p,p),\beta(A^+_p,p))+ 
\sum_{b\in B^1_{s}\setminus \{\alpha(A^+_p,p)\}}w(b,\beta(A^+_p,p))\\
&\ge& w(\alpha(A^+_p,p),\beta(A^+_p,p)) = 1.
\end{eqnarray*}
Note that $\lambda(\beta(A^+_p,p))=0$ by Definition~\ref{canonical lambda minus}. Thus,
$$
\lambda(\beta(A^+_p,p))\cdot s(\beta(A^+_p,p)) + \sum_{b\in B^1_{s}}w(b,\beta(A^+_p,p))\ge 1.
$$
Hence, by Definition~\ref{canonical theta minus},
$$
\lambda(\beta(A^+_p,p))\cdot s(\beta(A^+_p,p)) + \sum_{b\in B^1_{ s}}w(b,\beta(A^+_p,p))\ge \theta(\beta(A^+_p,p)).
$$
Thus, $\beta(A^+_p,p)\in B^2_{ s}$ by Definition~\ref{Ak}. Therefore, $\beta(A^+_p,p)\in B^*_{ s}$ by Definition~\ref{A*}.
\end{proof}

The next lemma states that if we do spend at least $p$ on preventive marketing to agent $\alpha(A^+_p,p)$, then this agent will never adopt the product.

\begin{lemma}\label{alpha is locked}
For every $(A^+_p,p)\in \mathbb{L}$ and every spending function $s$, if $s(\alpha(A^+_p,p))\ge p$, then $\alpha(A^+_p,p)\notin A^*_{ s}$.
\end{lemma}
\begin{proof}
By Definition~\ref{A*}, it suffices to show that $\alpha(A^+_p,p)\notin A^k_{ s}$ for each $k\ge 0$. We prove this statement by induction on $k$. 

\noindent{\em Base Case:} By Definition~\ref{Ak}, we have $A^0_{ s}=A$. At the same time, by Definition~\ref{Z}, $(A^+_p,p)\in \mathbb{L}$ implies that $A\subseteq \mathcal{A}_0$. Hence, $A^0_{ s}\subseteq \mathcal{A}_0$. Therefore, $\alpha(A^+_p,p)\notin A^0_{ s}$ by Definition~\ref{A}.

\noindent{\em Induction Step:} Suppose that $\alpha(A^+_p,p)\notin A^k_{s}$ and $\alpha(A^+_p,p)\in A^{k+1}_{s}$. 
Thus, by Definition~\ref{Ak},
 $$
 \lambda(\alpha(A^+_p,p))\cdot s(\alpha(A^+_p,p))+ \sum_{b\in A^k_{s}}w(b,\alpha(A^+_p,p))\ge \theta(\alpha(A^+_p,p)).
 $$
 Hence, by Definition~\ref{canonical lambda minus}, Definition~\ref{canonical w minus}, and Definition~\ref{canonical theta minus},
 $$
 1\cdot s(\alpha(A^+_p,p))+ 0\ge \epsilon - p.
 $$
 Thus, $s(\alpha(A^+_p,p))\le p-\epsilon$. Therefore, $s(\alpha(A^+_p,p))< p$ since $\epsilon>0$. This contradicts the assumption $s(\alpha(A^+_p,p))\ge p$ of the lemma.
\end{proof}

As we have seen in the previous lemma, spending at least $p$ on preventive marketing to agent $\alpha(A^+_p,p)$ prevents it from adopting the product. We now show that spending at least $p$ on agent $\alpha(A^+_p,p)$ prevents all agents in set $\mathcal{A}_0\setminus A^+_p$ from adopting the product. See Figure~\ref{second canonical network figure}.

\begin{lemma}\label{lock lemma}  
$A^*_{s}\cap\mathcal{A}_0\subseteq A^+_p$, where $(A^+_p,p)\in \mathbb{L}$ and $s$ is an arbitrary spending function such that $s(\alpha(A^+_p,p))\ge p$.
\end{lemma}
% \begin{lemma}[old]\label{old lock lemma}  
% $A^*_{s}\cap(\mathcal{A}_0\cup\{\alpha(A^+_p,p),\beta(A^+_p,p)\})\subseteq A^+_p$, where $(A^+_p,p)\in \mathbb{L}$ and $s$ be an arbitrary spending function such that $s(\alpha(A^+_p,p))\ge p$.
% \end{lemma}
\begin{proof}
Let $(A^+_p,p)\in \mathbb{L}$  and $s$ be an arbitrary spending function such that $s(\alpha(A^+_p,p))\ge p$. By Definition~\ref{A*}, it suffices to show that $A^k_{s}\cap\mathcal{A}_0\subseteq A^+_p$ for each $k\ge 0$. Instead, we prove the following two statements simultaneously by induction on $k$:
$$
\begin{cases}
A^k_{s}\cap\mathcal{A}_0\subseteq A^+_p,\\
\beta(A^+_p,p)\notin A^k_{s}.
\end{cases}
$$

\noindent{\em Base Case:} Suppose that $a\in A^0_{s}$. Thus, $a\in A$ by Definition~\ref{Ak}. Hence, $\vdash A\rhddash_p a$ by Reflexivity axiom. Therefore, $a\in A^+_p$ by Definition~\ref{A+}.
Assume now that $\beta(A^+_p,p)\in A^0_{s}$. Thus, $\beta(A^+_p,p)\in A\subseteq \mathcal{A}_0$ by Definition~\ref{Ak} and Definition~\ref{Z}, which is a contradiction with $\beta(A^+_p,p)\notin\mathcal{A}_0$ by the choice of $\alpha(\ell)$ and $\beta(\ell)$.

\noindent{\em Induction Step:} Assume that

$$
\begin{cases}
A^k_{s}\cap\mathcal{A}_0\subseteq A^+_p,\\
\beta(A^+_p,p)\notin A^k_{s}.
\end{cases}
$$
We need to show that
\begin{equation}\label{two statements}
\begin{cases}
A^{k+1}_{s}\cap\mathcal{A}_0\subseteq A^+_p,\\
\beta(A^+_p,p)\notin A^{k+1}_{s}.
\end{cases}
\end{equation}
To prove the first statement, suppose that there is $a\in \mathcal{A}_0$ such that $a\in A^{k+1}_{s}\setminus A^+_p$. Note that 
$a\notin A^{k}_{s}$ by the induction hypothesis. Thus, by Definition~\ref{Ak},
 $$
 \lambda(a)\cdot s(a)+ \sum_{b\in A^k_{s}}w(b,a)\ge \theta(a).
 $$
 Hence, by Definition~\ref{canonical lambda minus} and Definition~\ref{canonical theta minus}, 
 $$
 0\cdot s(a)+ \sum_{b\in A^k_{s}}w(b,a)\ge |\{(A^+_q,q)\in \mathbb{L}\;|\; a\in \mathcal{A}_0\setminus A^+_q\}|.
 $$
 Therefore, by Definition~\ref{canonical w minus}, 
 $$
 \sum_{\ell\in \{(A^+_q,q)\in \mathbb{L}\;|\; a\in \mathcal{A}_0\setminus A^+_q, \beta(\ell)\in A^k_{s}\}}w(\beta(\ell),a)\ge |\{(A^+_q,q)\in \mathbb{L}\;|\; a\in \mathcal{A}_0\setminus A^+_q\}|.
 $$
 Hence, again by Definition~\ref{canonical w minus}, 
 $$
 \sum_{\ell\in \{(A^+_q,q)\in \mathbb{L}\;|\; a\in \mathcal{A}_0\setminus A^+_q, \beta(\ell)\in A^k_{s}\}}1\ge |\{(A^+_q,q)\in \mathbb{L}\;|\; a\in \mathcal{A}_0\setminus A^+_q\}|.
 $$
 Thus,
 $$
|\{(A^+_q,q)\in \mathbb{L}\;|\; a\in \mathcal{A}_0\setminus A^+_q, \beta(\ell)\in A^k_{s}\}|\ge 
|\{(A^+_q,q)\in \mathbb{L}\;|\; a\in \mathcal{A}_0\setminus A^+_q\}|.
 $$
 At the same time,
 $$
 \{(A^+_q,q)\in \mathbb{L}\;|\; a\in \mathcal{A}_0\setminus A^+_q, \beta(\ell)\in A^k_{s}\}\subseteq \{(A^+_q,q)\in \mathbb{L}\;|\; a\in \mathcal{A}_0\setminus A^+_q\}.
 $$
 Then it must be the case that
 $$
 \{(A^+_q,q)\in \mathbb{L}\;|\; a\in \mathcal{A}_0\setminus A^+_q, \beta(\ell)\in A^k_{s}\}=\{(A^+_q,q)\in \mathbb{L}\;|\; a\in \mathcal{A}_0\setminus A^+_q\}.
 $$
 Hence, $\beta(\ell)\in A^k_{s}$ for all $\ell\in \{(A^+_q,q)\in \mathbb{L}\;|\; a\in \mathcal{A}_0\setminus A^+_q\}$. In particular, $\beta(A^+_p,p)\in A^{k}_{s}$, which is a contradiction to the induction hypothesis.
 
 To prove the second statement from (\ref{two statements}), suppose that $\beta(A^+_p,p)\in A^{k+1}_{s}$. Note that $\beta(A^+_p,p)\notin A^{k}_{s}$ due to the induction hypothesis. Thus, by Definition~\ref{Ak},
 $$
 \lambda(\beta(A^+_p,p))\cdot s(\beta(A^+_p,p))+ \sum_{b\in A^k_{s}}w(b,\beta(A^+_p,p))\ge \theta(\beta(A^+_p,p)).
 $$
 Hence, by Definition~\ref{canonical lambda minus} and Definition~\ref{canonical theta minus}, 
 $$
 0\cdot s(\beta(A^+_p,p))+ \sum_{b\in A^k_{s}}w(b,\beta(A^+_p,p))\ge 1.
 $$ 
 Thus, there must exist at least one $b\in A^k_{s}$ such that $w(b,\beta(A^+_p,p))>0$. By Definition~\ref{canonical w minus} and Lemma~\ref{alpha is locked}, this implies that $b\in \mathcal{A}_0\setminus  A^+_p$, which is a contradiction to the first part of the induction hypothesis, i.e. $A^k_s\cap\mathcal{A}_0\subseteq A^+_p$. 
\end{proof}

\begin{lemma}\label{preventive base 1}
For each $B,C\subseteq \mathcal{A}_0$ and each $q\in P$,
if $B\rhddash_q C\in X$, then $N\vDash B\rhddash_q C$.
\end{lemma}
\begin{proof}
Consider any spending function $s$ such that $\|s\|\le q$. By Definition~\ref{preventive sat}, it suffices to show that $C\subseteq B^*_{s}$.
Suppose that there is $c_0\in C$ such that $c_0\notin B^*_{s}$.

Thus, by Lemma~\ref{or lemma}, there exists $(A^+_p,p)\in \mathbb{L}$ such that $c_0\notin A^+_p$ and $\beta(A^+_p,p)\notin B^*_{s}$. The latter, by Lemma~\ref{if rhs then beta}, implies that $B^*_{s}\cap(\mathcal{A}_0\setminus A^+_p)=\varnothing$. Hence, $(B^*_{s}\cap\mathcal{A}_0)\setminus A^+_p=\varnothing$. Then, $B^*_{s}\cap \mathcal{A}_0\subseteq A^+_p$. Thus, $B\subseteq B^*_{s}\cap \mathcal{A}_0\subseteq A^+_p$ by Definition~\ref{Ak} and Definition~\ref{A*}. We next consider the following two cases:

\noindent{\em Case I:} $p\le q$. In this case, assumption $B\rhddash_q C\in X$, by Monotonicity axiom, implies that $X\vdash B\rhddash_p C$. At the same time,  by Reflexivity axiom, $B\subseteq A^+_p$ implies that $\vdash A^+_p\rhddash_p B$. Thus, $X\vdash A^+_p\rhddash_p C$ by Transitivity axiom. Again by Reflexivity axiom, we have $\vdash C\rhddash_p c_0$. Hence, $X\vdash A^+_p\rhddash_p c_0$ by Transitivity axiom. Thus, $X\vdash A\rhddash_p c_0$ by Lemma~\ref{ArhdA+} and Transitivity axiom. Therefore, $c_0\in A^+_p$, which is a contradiction with the choice of set $A$.

\noindent{\em Case II:} $p > q$. Then, $p-\epsilon> q$ by the choice of $\epsilon$. Hence, 
$$s(\alpha(A^+_p,p))\le \|s\|\le q< p-\epsilon.$$
Therefore, $\beta(A^+_p,p)\in B^*_{s}$ by Lemma~\ref{if alpha then beta}, which is a contradiction with the choice of set $A$.
\end{proof}

\begin{lemma}\label{preventive base 2}
For each $B,C\subseteq \mathcal{A}_0$ and each $q\in P$,
if $N\vDash B\rhddash_q C$,
then  $B\rhddash_q C\in X$. 
\end{lemma}
\begin{proof}
Suppose that $B\rhddash_q C\notin X$. Thus, by Lemma~\ref{rhd set} and the maximality of set $X$, there is $c_0\in C$ such that $X\nvdash B\rhddash_q c_0$. Hence, $c_0\notin B^+_q$ by Definition~\ref{A+}. Consider spending function $s$ such that
$$
s(a)=
\begin{cases}
q, & \mbox{if $a=\alpha(B^+_q,q)$},\\
0, & \mbox{otherwise}.
\end{cases}
$$
Note that $\|s\|=q$. Thus, $C\subseteq B^*_{s}$ by the assumption $N\vDash B\rhddash_q C$ of the lemma. Hence, $c_0\in B^*_{s}$. This together with $c_0\notin B^+_q$ contradicts with Lemma~\ref{lock lemma} and $c_0\in C\subseteq \mathcal{A}_0$.
\end{proof}

%We are now getting close to finishing the of Theorem~\ref{preventive completeness}. Recall that we have supposed that $\nvdash\phi$ for some formula $\phi\in\Phi(\mathcal{A})$.  

\begin{lemma}\label{preventive main induction}
$\psi\in X$ iff $N\vDash\psi$, for each $\psi\in\Phi(\mathcal{A}_0)$. 
\end{lemma}
\begin{proof}
We prove the lemma by induction on the structural complexity of formula $\psi$. The base case follows from Lemma~\ref{preventive base 1} and Lemma~\ref{preventive base 2}. The induction step follows from Definition~\ref{preventive sat} and maximality and consistency of set $X$ in the standard way.
\end{proof}
To finish the proof of Theorem~\ref{preventive completeness} note that $\neg\phi\in X$ due to the choice of the set $X$. Thus, $\phi\notin X$ due to consistency of set $X$. Therefore, $N\nvDash\phi$ by Lemma~\ref{preventive main induction}.

\section{Conclusion}\label{conclusion section}

In this paper we have suggested a way of adding marketing to the standard threshold model of diffusion in social networks. The model is general enough to simulate both promotional and preventive marketing. We have also defined formal logical systems for reasoning about influence relation in social networks with marketing of these two types. Both systems are based on Armstrong's axioms from the database theory. The main technical results of the paper are the completeness theorems for these two systems. A possible extension of this work is an analysis of the computational complexity of this model.

\bibliography{sp}

\end{document}